\documentclass[aps,pra,twocolumn,floatfix,superscriptaddress,nofootinbib]{revtex4-1}
\usepackage{amsmath,amssymb,amsfonts,bbm,graphicx,color,mathtools}
\usepackage{hyperref}
\usepackage[utf8]{inputenc}
\usepackage[T1]{fontenc}
\usepackage{physics}
\usepackage{pdfcomment}
\usepackage{times}
\usepackage{amsthm}
\usepackage{subfigure}

\newtheorem{propx}{Proposition}

\newtheorem{lemma}{Corollary}


\usepackage{enumitem}

\begin{document}
\title{{Symmetry and} block structure of the Liouvillian superoperator in 
partial secular approximation}
\author{Marco Cattaneo}
\affiliation{Instituto de F\'{i}sica Interdisciplinar y Sistemas Complejos IFISC 
(CSIC-UIB),
Campus Universitat Illes Balears, E-07122 Palma de Mallorca, Spain}
\affiliation{QTF Centre of Excellence, Turku Centre for Quantum Physics, 
Department of
Physics and Astronomy, University of Turku, FI-20014 Turun Yliopisto, Finland}
\author{Gian Luca Giorgi}
\affiliation{Instituto de F\'{i}sica Interdisciplinar y Sistemas Complejos IFISC 
(CSIC-UIB),
Campus Universitat Illes Balears, E-07122 Palma de Mallorca, Spain}
\author{Sabrina Maniscalco}
\affiliation{QTF Centre of Excellence, Turku Centre for Quantum Physics, 
Department of
Physics and Astronomy, University of Turku, FI-20014 Turun Yliopisto, Finland}
\affiliation{QTF Centre of Excellence, Department of Applied Physics, School of 
Science, Aalto University, FI-00076 Aalto, Finland}
\author{Roberta Zambrini}
\affiliation{Instituto de F\'{i}sica Interdisciplinar y Sistemas Complejos IFISC 
(CSIC-UIB),
Campus Universitat Illes Balears, E-07122 Palma de Mallorca, Spain}
\date{\today}
\begin{abstract}
We address the structure of the Liouvillian superoperator for a broad class of 
bosonic and fermionic {Markovian open} systems interacting with {stationary} 
environments. We show that the {accurate} application of
the partial secular approximation in the derivation of the Bloch-Redfield 
master equation naturally induces a symmetry on the superoperator level, which 
may greatly reduce the complexity of the master equation by decomposing the 
Liouvillian superoperator into independent blocks. Moreover, we prove that, 
{if the steady state of the system is unique,}
{one} single block contains all the information about {it}, and that this 
imposes a constraint on the possible steady-state coherences of the 
{unique} state, ruling out some of them. To provide some examples, we show 
how the symmetry appears for two coupled spins interacting with separate baths, 
as well as for two harmonic oscillators immersed in a common environment. In 
both cases the standard derivation and solution of the master equation {is} 
simplified, as well as the search for the steady state. {The 
block-diagonalization may not appear} when a local master equation is chosen.
\end{abstract}
\maketitle
\section{Introduction}
\label{sec:intro}
Open quantum systems are nowadays a well-established {framework} whose 
{theoretical}
aspects have been investigated in depth 
\cite{BreuerPetruccione,Weiss,Rivas2012}. 
{An important branch is {represented by} Markovian open systems 
\cite{Rivas2010a,RevModPhys.88.021002} and quantum dynamical 
semigroups \cite{Alicki2007}.
In particular, the {generator} of a quantum dynamical semigroup is a 
time-independent} 
\textit{Liouvillian superoperator} $\mathcal{L}$, such that if $\rho_S(0)$ is 
the initial state of the system, the state at time $t$ is given by 
$\rho_S(t)=\exp(\mathcal{L} t)[\rho_S(0)]$. 
The solution of the master equation {providing the dynamics} of the system 
relies on finding the Liouvillian $\mathcal{L}$. 

{The evolution of a dynamical semigroup is described by a master equation in} 
the so-called Gorini–Kossakowski–Sudarshan–Lindblad (GKLS) form \cite{Gorini1976a,Lindblad1976,Chruscinski2017a}. 
This form has been extensively studied {during} the recent years 
\cite{Baumgartner2007,Baumgartner2008,Prosen2008,Fagnola2008,Prosen2010,
Prosen2010b,baumgartner2012structures,PhysRevX.6.041031,PhysRevE.95.042137,
sa2019spectral}, 
with a particular attention on the steady state structure, given the importance 
of, for instance, steady-state coherences in quantum thermodynamics 
\cite{PhysRevLett.121.070401,perez2018endurance,PhysRevA.99.042320} or of 
information-preserving steady states \cite{PhysRevA.82.062306}. The form of the steady state is also crucial to understand the process of quantum thermalization \cite{Ostilli2017}.
The investigation of the role of symmetry in the semigroup evolution has been 
very active as well 
\cite{Buca2012,PhysRevB.90.125138,Albert2014,PhysRevA.94.052129,
PhysRevA.98.063815,PhysRevA.97.052106,manzano2018harnessing,SanchezMunoz2019,Styliaris2019}. {Given the 
simplicity of the GKLS master 
equation}, a thornier issue has been to characterize which microscopic physical 
models {of systems and environments lead to a reduced system} evolution 
described 
by {this} master equation. In 1965 Redfield derived a 
Markovian master equation by assuming weak coupling between system and 
environment \textit{and} making some considerations about the relevant 
timescales of the evolution 
\cite{Redfield1965}. This derivation and the subsequent Bloch-Redfield master 
equation are still commonly employed nowadays 
\cite{BreuerPetruccione,Weiss,jeske2015bloch}. 
A more formal derivation has been provided by Davies 
\cite{Davies1974,Davies1976}, showing that the semigroup evolution is perfectly 
recovered when the coupling between 
system and environment is infinitesimally small. {In some situations, e.g. in 
the case of two slightly-detuned spins \cite{Benatti2010,Cattaneo2019}, the Davies' limit 
cannot be 
performed, since it corresponds} to applying a ``full secular approximation'' 
{removing all the oscillating terms in the 
{interaction picture} dynamics without discriminating which of them are fast and 
which are slow, instead of a {more accurate} partial secular approximation. The 
latter} 
was implicitly suggested by 
Redfield himself \cite{Redfield1965}, and an extensive study about it has been 
performed in the very recent past 
\cite{jeske2015bloch,Cresser2017a,Farina2019,Hartmann2020,Cattaneo2019,mccauley2019completely}, 
in particular showing that applying an accurate partial secular approximation 
to the 
microscopic derivation of the master equation allows one to recover the GKLS 
form \cite{Cresser2017a,Farina2019}.

In this work we show how a symmetry on the superoperator level arises due to the 
partial secular approximation. Our discussion is valid for a broad class of 
systems, 
{that can be recast as $M$} non-interacting fermionic or bosonic modes weakly 
coupled to {stationary} Markovian environments. 
The symmetry consists in the invariance of the Liouvillian superoperator under 
the action of the total-number-of-particles superoperator. 
We stress that the symmetry is on the superoperator level, i.e. it is not a 
symmetry of the system Hamiltonian, but of the {full} master equation of the 
open system. 
Following the formalism discussed in Ref.~\cite{Albert2014}, we can exploit it 
to block-diagonalize the Liouvillian superoperator and greatly reduce 
the complexity of the master equation. Complexity reduction of abstract GKLS 
master equations in fermionic or bosonic systems has been {also} 
addressed in 
some extensive works by Prosen et al. \cite{Prosen2008,Prosen2010,Prosen2010b}, {while Torres exploited a symmetry of the Hamiltonian to find the solution of master equations without gain \cite{Torres2014}.}
{Once having exploited the symmetry to obtain the block-diagonalization}, we 
observe that,
{if the steady state of the system is unique}, {one} single block contains 
all the information about {it}. 
This not only helps to find it, but also imposes a constraint on the 
{corresponding} steady-state coherences. Curiously, {the symmetry 
arises} only when considering a global master equation, 
while it may not be valid anymore when using a local one 
\cite{Gonzalez2017,Hofer2017,Cattaneo2019}.

We review the derivation of the Bloch-Redfield master equation and subsequent 
partial secular approximation in Sec.~\ref{sec:prelim}, as well as the theory of 
symmetries and conserved quantities in Lindblad master equations. 
Sec.~\ref{sec:block} is devoted to the discussion of the symmetry on the 
superoperator level and to the block-diagonalization of the Liouvillian. In 
particular, Sec.~\ref{sec:requirements} discusses the class of system for which 
our analysis is valid, while Sec.~\ref{sec:symRes} presents the main result and Sec.~\ref{sec:consequences}
its consequences. We provide some illustrative examples of the action of the 
symmetry in Sec.~\ref{sec:examples}, distinguishing between fermionic and 
bosonic scenarios. Finally, we conclude in Sec.~\ref{sec:conclusions} with a 
discussion about our results.

\section{Formal framework} 
\label{sec:prelim}

\subsection{Markovian master equations with partial secular approximation}
\label{sec:partial}
Let us consider an open quantum system $S$ with associated Hilbert space 
$\mathsf{H_S}$ of dimension $N$, described at time $t$ by the $N\times N$ 
density 
matrix $\rho_S(t)$. $S$ is coupled to an external environment $E$ through the 
interaction Hamiltonian $\hat{H}_I$, and 
{throughout the work we restrict ourselves to stationary environments. The} 
{full system-bath} Hamiltonian can be written as:
\begin{equation}
\label{eqn:Hamiltonian}
\begin{split}
\hat{H}=&\hat{H}_S+\hat{H}_E+\hat{H}_I\\
=&\hat{H}_S+\hat{H}_E+\mu\sum_\alpha \hat{A}_\alpha\otimes \hat{B}_\alpha,
\end{split}
\end{equation}
wherer $\hat{H}_S$ is the free Hamiltonian of the system, $\hat{H}_E$ is the 
free Hamiltonian of the environment, $\hat{A}_\alpha$ are system 
operators while $\hat{B}_\alpha$ are bath operators. $\mu$ is a coupling 
constant with units of energy, and in the weak-coupling limit considered here 
we 
assume $\mu$ far smaller than the other characteristic energies of the system. 
We set $\hbar=1$, so that the units of measure of time are 
$[time]=[energy]^{-1}$. 

We term $\ket{e_n}$ the eigenvectors of the free Hamiltonian of the system, 
which may be degenerate as well, such that $\hat{H}_S=\sum_n \epsilon_n 
\ket{e_n}\bra{e_n}$. 
The \textit{jump operators} of the system are defined as 
\cite{BreuerPetruccione}
\begin{equation}
\label{eqn:jumpOp}
\hat{A}_\alpha(\omega)=\sum_{\epsilon_m-\epsilon_n=\omega}\ket{e_n}\bra{e_n}
\!\hat{A}_\alpha\!\ket{e_m}\bra{e_m}.
\end{equation}

We assume that the open system $S$ follows a Markovian, non-unitary evolution 
due to the coupling to the {stationary} environment $E$. {The master equation 
describing a time-independent dynamical semigroup is written as:}
\begin{equation}
\label{eqn:masterEqLiouvillian}
\frac{d}{dt}\rho_S(t)=\mathcal{L}[\rho_S(t)],
\end{equation} 
where $\mathcal{L}$ is the \textit{Liouvillian superoperator} acting on the 
$N^2$-dimensional Hilbert space $\mathsf{L}$ of the linear operators on 
$\mathsf{H}_S$, called Liouville space \cite{Albert2014}, which contains the 
convex subset of the density matrices. {In particular, for the Bloch-Redfield 
master 
equation in partial secular approximation {(PSA)} \cite{BreuerPetruccione,Weiss,Cattaneo2019}:}
\begin{equation}
\label{eqn:Liouvillian}
\mathcal{L}=-i[\hat{H}_S+\hat{H}_{LS},\,\cdot\,]+\mathcal{D}[\,\cdot\,],
\end{equation}
where $\hat{H}_{LS}$ is the 
\textit{Lamb-shift Hamiltonian} given by:
\begin{equation}
\label{eqn:lambShift}
\hat{H}_{LS}=\sum_{\alpha,\beta}\sum_{(\omega,\omega')\in\mathsf{PSA}} 
S_{\alpha\beta}(\omega,\omega')\hat{A}_\alpha^\dagger(\omega') 
\hat{A}_\beta(\omega),
\end{equation}
while the \textit{dissipator} reads
\begin{equation}
\label{eqn:dissipator}
\begin{split}
\mathcal{D}[\rho_S]=&\sum_{\alpha,\beta}\sum_{(\omega,\omega')\in\mathsf{PSA}} 
\gamma_{\alpha\beta}(\omega,\omega')\Big(\hat{A}_\beta(\omega)\rho_S 
\hat{A}_\alpha^\dagger(\omega')\\
&-\frac{1}{2}\{\hat{A}_\alpha^\dagger(\omega')\hat{A}_\beta(\omega),\rho_S\}
\Big).
\end{split}
\end{equation}
$S_{\alpha\beta}(\omega,\omega')$ and $\gamma_{\alpha\beta}(\omega,\omega')$ are 
functions of the autocorrelation functions of the bath operators 
$B_\alpha$\footnote{We refer the reader to Ref.~\cite{Cattaneo2019} for their 
precise form.}. The PSA removes all the terms in the summation with frequencies 
$\omega$ and $\omega'$ such that
\begin{equation}
\label{eqn:partialSecular}
\exists\, t^*\textnormal{ such that } \abs{\omega-\omega'}^{-1}\ll t^* \ll 
\tau_R,
\end{equation}
where $\tau_R$ is the relaxation time of the system, i.e. the time in which 
$\rho_S$ approaches the dynamical equilibrium 
\cite{BreuerPetruccione,Cattaneo2019}. {We can express Eq.~\eqref{eqn:partialSecular} as $\omega-\omega'\neq \mathcal{O}_{t^*}(\tau_R^{-1})$, where for convenience we introduce the notation $\mathcal{O}_{t^*}$, defined as:
\begin{equation}
\label{eqn:ordering}
x=\mathcal{O}_{t^*}(y)\textnormal{ if }\nexists\, t^*\textnormal{ such that } x^{-1}\ll t^* \ll y^{-1}.
\end{equation}}
{In the weak-coupling limit considered here we have}
\begin{equation}
\label{eqn:tauR}
\tau_R=O(\mu^{-2}),
\end{equation}
being the master equation of the second order in $\mu$ \cite{Rivas2012}. 
{Appendix~\ref{sec:blochLin} discusses {why} Eq.~\eqref{eqn:masterEqLiouvillian} with Liouvillian in Eq.~\eqref{eqn:Liouvillian} can be recast in the GKLS form:
\begin{equation}
\label{eqn:LindbladForm}
\begin{split}
\mathcal{L}[\rho_S(t)]=&-i[\hat{H'},\rho_S(t)]\\
&+\sum_{l=1}^{N^2-1} \hat{F}_l \rho_S(t) \hat{F}_l^\dagger -\frac{1}{2}\{\hat{F}_l^\dagger \hat{F}_l, \rho_S(t)\},
\end{split}
\end{equation}
where $\hat{H'}=\hat{H'}^\dagger$ is the effective Hamiltonian including the Lamb shift, and $\{\hat{F}_l\}_{l=1}^{N^2-1}$ are the \textit{Lindblad operators} \cite{BreuerPetruccione}.}

From now on, we will use calligraphic letters (such as $\mathcal{L}$) to 
indicate superoperators \textit{acting on} $\mathsf{L}$, while we will use 
capital letters, with 
hats {when needed to avoid confusion}, (such as $\hat{H}$) for operators 
\textit{living in} $\mathsf{L}$, which for instance may act on the Hilbert space 
of the system 
$\mathsf{H_S}$. The density matrices $\rho_S$ are elements of $\mathsf{L}$ as 
well. Appendix~\ref{sec:superoperatorsSpace} discusses the language of 
superoperators in more 
detail.

\subsection{Symmetries and conserved quantities in the Lindblad formalism}
\label{sec:symmetries}
In this section we {introduce the concepts of symmetries and conserved 
quantities in the Lindblad formalism following the recent work by Albert and 
Jiang \cite{Albert2014}}. 
{Let us} assume that the Lindblad evolution of an open system $S$ is described by the Liouvillian 
superoperator 
$\mathcal{L}$ as discussed in 
Sec.~\ref{sec:partial}. {Given an observable $\hat{J}=\hat{J}^\dagger$} acting 
on $\mathsf{H_S}$ 
and living in $\mathsf{L}$, we have the following definitions:
\begin{itemize}
\item $\hat{J}$ is a \textit{conserved quantity} if it is a constant of motion 
under the non-unitary evolution generated by the master equation, i.e. if 
$\mathcal{L}^\dagger[\hat{J}(t)]=0$ for all $t$.
\end{itemize}
We construct the one-parameter unitary group whose elements are 
$\hat{U}_\phi=\exp(i\phi \hat{J})$ with $\phi\in\mathbb{R}$, and then we define 
the associated superoperators $\mathcal{U}_\phi$ as 
$\mathcal{U}_\phi^\dagger[\hat{O}]=\hat{U}_\phi^\dagger \hat{O} \hat{U}_\phi$, 
with $\hat{O}\in \mathsf{L}$. We can analogously write 
$\mathcal{U}_\phi=\exp(i\phi\mathcal{J})$, {where $\mathcal{J}$ is the 
superoperator associated to $\hat{J}$ through $\mathcal{J}=[\hat{J},\cdot\,]$. In the language of the isomorphism introduced through the tensor product notation in 
Appendix~\ref{sec:superoperatorsSpace}, we have
$\mathcal{J}=\hat{J}\otimes\mathbb{I}_N-\mathbb{I}_N\otimes \hat{J}^T$. }
\begin{itemize}
\item $\hat{J}$ generates \textit{a continuous symmetry on the superoperator 
level} if $\mathcal{U}_\phi^\dagger\mathcal{L}\mathcal{U}_\phi=\mathcal{L}$ for 
all $\phi$, or equivalently $[\mathcal{J},\mathcal{L}]=0$. {The continuous symmetry is also called \textit{covariance} \cite{Holevo1993,Holevo1996,vacchini2010covariant}, given that it corresponds to the equivalence $\hat{U}_\phi^\dagger \mathcal{L}[\rho_S]\hat{U}_\phi=\mathcal{L}[\hat{U}_\phi^\dagger \rho_S \hat{U}_\phi]$, for any state of the system $\rho_S$.} 
\end{itemize}

If the evolution of the system were unitary and driven only by the Hamiltonian 
$\hat{H}_S$, according to Noether's theorem a conserved quantity would always 
generate a
symmetry and viceversa. In the framework of open systems this is no longer true, 
since for instance a symmetry on the superoperator level not always implies a 
symmetry 
on the operator level. {In particular, if the master equation is in the Lindblad form as in 
Eq.~\eqref{eqn:LindbladForm}, we can consider the following three propositions:}
\begin{enumerate}[label=(\roman*)]
\item $[\hat{J},\hat{H}']=[\hat{J},\hat{F}_l]=0 \quad\forall\,l,$
\item $\frac{d}{dt}\hat{J}(t)=\mathcal{L}^\dagger[\hat{J}(t)]=0,$
\item $\mathcal{U}_\phi^\dagger\mathcal{L}\mathcal{U}_\phi=\mathcal{L}\;\forall 
\phi\in\mathbb{R},$ or equivalently $[\mathcal{J},\mathcal{L}]=0$.
\end{enumerate}
Then, we have that (i) implies (ii) and (iii), but no other logical implications 
are present \cite{Albert2014}. 
This tells us that in order for an observable $\hat{J}$ to both be a conserved 
quantity {(ii)} and generate a symmetry  {(iii)},
it needs to commute both with the Hamiltonian $\hat{H}'$ driving the unitary 
part of the evolution and with each Lindblad operator.

{For the purpose of this paper, we are interested in the observable representing 
the total number of particles in a system:} 
suppose we have a system  
of $M$ bosonic or fermionic modes; then, the Hilbert space of the system is the 
tensor product of the Hilbert spaces of the $M$ modes. 
{The total-number-of-particles operator reads
\begin{equation}
\label{eqn:totalNumberOp}
\hat{N}=\sum_{k=1}^M \hat{n}_k,
\end{equation}
where $\hat{n}_k$ is the particle number operator of the $k$-th mode. $\hat{N}$ 
generates the one-parameter group $\hat{U}_\phi=\exp(i\phi \hat{N})$. 
If we set $\phi=\pi$, we obtain the \textit{parity operator}:
\begin{equation}
\label{eqn:parityOp}
\hat{P}=\exp(i\pi \hat{N}).
\end{equation}
The parity operator satisfies the properties $\hat{P}^2=\mathbb{I}$ and 
$\hat{P}^\dagger=\hat{P}$, and as a consequence it only has two eigenvalues, 
$\pm 1$. 
{Parity} is an observable which can generate a \textit{discrete} symmetry on 
the superoperator level\footnote{Discrete simmetries in the Lindblad formalism 
deserve a separate discussion, and we refer the interested reader to 
Ref.~\cite{Albert2014}.}. In analogy with the definition of a continuous 
symmetry, we write the parity superoperator as 
\begin{equation}
\label{eqn:paritySupOp}
\mathcal{P}=\exp(i\pi\mathcal{N}),
\end{equation}
{where $\mathcal{N}$ is defined as
\begin{equation}
\label{eqn:supOpAbstr}
\mathcal{N}=[\hat{N},\cdot\,].
\end{equation}
Equivalently, using the tensor product notation {(see 
Appendix~\ref{sec:superoperatorsSpace})} we have
\begin{equation}
\label{eqn:superoperatorN}
\mathcal{N}=\hat{N}\otimes\mathbb{I}-\mathbb{I}\otimes\hat{N}^T.
\end{equation}	}}

Being different objects, symmetries and conserved quantities play a different 
role in the analysis of the evolution of open quantum systems 
\cite{Baumgartner2007,Baumgartner2008,Buca2012,Albert2014}. Conserved quantities 
are of fundamental importance to identify the structure of the space of 
stationary states of the systems \cite{Albert2014}, related to the problem of 
finding decoherence-free subspaces \cite{Lidar2003}. Symmetries can help in 
simplifying the form of the Liouvillian superoperator, and thus in solving the 
master equation. Indeed, if we identify a symmetry such that 
$[\mathcal{J},\mathcal{L}]=0$, we can block-diagonalize the Liouvillian with 
each block labeled by a different eigenvalue of $\mathcal{J}$. As we will see in 
the next section, this can greatly reduce the complexity of the master equation.

\section{The block structure of the Liouvillian in partial secular 
approximation}
\label{sec:block}
In this section we will show how, for a broad class of {models}, the partial 
secular approximation naturally induces a symmetry on the superoperator level, 
which can be exploited to simplify the master equation. Note that we can apply 
the concepts of Sec.~\ref{sec:symmetries}, introduced in the Lindblad formalism, 
to the Bloch-Redfield master equation in partial secular approximation, {since 
as explained in Appendix~\ref{sec:blochLin}} the latter can be brought to 
the GKLS form. 

{We start by introducing the suitable class of Hamiltonians in 
Sec.~\ref{sec:requirements}, and then we focus on the identification of the 
symmetry in 
Sec.~\ref{sec:symRes}. {Section}~\ref{sec:consequences} discusses a series of 
interesting applications and consequences of the main result.}

\subsection{Delimiting the suitable class of systems}
\label{sec:requirements}
{Our analysis applies to all systems that can be cast} as the sum of the free 
Hamiltonians of $M$ non-interacting bosonic or fermionic modes{, with
\begin{equation}
\label{eqn:HamiltonianNonIntModes}
\hat{H}_S=\sum_{k=1}^M E_k \hat{c}_k^\dagger\hat{c}_k=\sum_{k=1}^M E_k 
\hat{n}_k,
\end{equation}
and} $E_k$ is the energy quantum of the $k$-th mode.

Eq.~\eqref{eqn:HamiltonianNonIntModes} describes a broad class of Hamiltonians 
which are particularly relevant in the field{s} of condensed matter {and optical 
physics}. {For instance, any quadratic Hamiltonian, that is to say any Hamiltonian of the form
\begin{equation}
\label{eqn:quadraticHam}
\begin{split}
\hat{H}_S
=&\sum_{j,k=1}^M 
\left(\alpha_{jk}\hat{a}_j^\dagger\hat{a}_k+\beta_{jk}\hat{a}_j\hat{a}
_k+h.c.\right),
\end{split}
\end{equation}
where $\hat{a}_j$ is an annihiliation operator, can be rewritten as a sum of 
non-interacting modes as in Eq.~\eqref{eqn:HamiltonianNonIntModes} 
\cite{Lieb1961,Nielsen2005}.} This is just a sufficient but 
not necessary condition, since more complex Hamiltonians may be {taken into} 
the form of 
Eq.~\eqref{eqn:quadraticHam}. In the case of bosons, {all $\hat{H}_S$ preserving 
Gaussian states can be recast as Eq.~\eqref{eqn:HamiltonianNonIntModes}.
These Hamiltonians contain linear and/or bilinear terms and can be {reduced 
into} the form of Eq.~\eqref{eqn:quadraticHam} through displacement 
transformations.}
Systems of uncoupled spins can be trivially seen as non-interacting fermions 
via Jordan–Wigner transformation \cite{Lieb1961,coleman2015introduction}, and thus are suitable for 
our discussion. {The same holds for interacting spin chains in which the total number of spin excitations is conserved (see Appendix~\ref{sec:JordanWigner}).} Two coupled qubits can be transformed into 
free fermions as well, as discussed in Appendix~\ref{sec:coupledSpinsFree}, {while} 
extensions to wider systems of interacting spins {(such as the Heisenberg model)} are tricky and must be 
considered 
case by case.

We now {set the relevant assumptions} on 
the interaction Hamiltonian $\hat{H}_I$ in Eq.~\eqref{eqn:Hamiltonian}. 
First of all, recalling that $\mu$ is the system-bath coupling constant defined 
in Eq.~\eqref{eqn:Hamiltonian}, {we set $E_k\neq\mathcal{O}_{t^*}(\mu^2)\;\forall\, k$, where we have used the notation introduced in Eq.~\eqref{eqn:ordering}}. Then, the 
interaction Hamiltonian is suitable for our analysis if \textit{at least one} of 
the following conditions holds:
\begin{itemize}
\item \textit{Condition I.} Each system operator $\hat{A}_\alpha$ in 
Eq.~\eqref{eqn:Hamiltonian} {involves} only single 
excitations, that is to say, each $\hat{A}_\alpha$ is a first-degree polynomial 
in the creation and annihilation operators $\hat{c}_k$. For instance, 
$\hat{A}_{\alpha'}=\hat{c}_1+\hat{c}_2^\dagger$ is a valid system operator, 
while $\hat{A}_{\alpha'}=\hat{c}_1\hat{c}_2^\dagger$ or 
$\hat{A}_{\alpha'}=\hat{c}_1\hat{c}_2$ are not.
\item \textit{Condition II.} Let us consider the set of energies 
$K=\{E_k\}_{k=1}^M$. Create two new sets by randomly selecting some elements of 
$K$, that can be repeated as well, and term them $X$ and $Y$; assume that they 
have different cardinality (number of elements): $\abs{X}\neq\abs{Y}$. {Then, we 
{exclude situations such} that $\sum_{E_m \in X} E_m=\sum_{E_l \in Y} 
E_l+\mathcal{O}_{t^*}(\mu^2)$}. 
{This condition can be relaxed depending on the structure of the system 
operators in the interaction Hamiltonian, as we will show in the proof in Appendix~\ref{sec:proofT1}.} {\textit{Condition II} relaxed as explained in Appendix~\ref{sec:proofT1} comprises \textit{Condition I} together with the assumption $E_k\neq\mathcal{O}_{t^*}(\mu^2)$.}
\end{itemize}

\subsection{The symmetry of the partial secular approximation}
\label{sec:symRes}
We will now show that, if the requirements of Sec.~\ref{sec:requirements} are 
satisfied, the number superoperator $\mathcal{N}$ defined in 
Eq.~\eqref{eqn:superoperatorN} commutes with the Liouvillian in partial secular 
approximation Eq.~\eqref{eqn:Liouvillian}, and therefore generates a symmetry on 
the superoperator level. 
\begin{propx}[Symmetry]
Let $\mathcal{L}$ be the Liouvillian superoperator describing the Markovian 
evolution of a quantum system that can be written as a collection of bosonic or 
fermionic non-interacting modes. If $\mathcal{L}$ has been derived, starting 
from the microscopic model of system+environment, through the Bloch-Redfield 
master equation in partial secular approximation, then it commutes with the 
number superoperator:
\begin{equation}
\label{eqn:symmetry}
[\mathcal{N},\mathcal{L}]=0,
\end{equation}
provided that one {of the conditions I or II} on the interaction 
Hamiltonian $\hat{H}_I$ discussed in Sec.~\ref{sec:requirements} holds.
\end{propx}
\begin{proof}
In Appendix~\ref{sec:proofT1}.
\end{proof}

{Note that Proposition 1 may be considered as the extension to systems of $M$ modes of the concept of \textit{phase-covariant master equation} \cite{vacchini2010covariant,Smirne2016,Teittinen2018,haase2019non}. By now, the latter has been addressed as the problem in which a system of a single qubit follows an open dynamics described by the Liouvillian $\mathcal{L}$ which is covariant under a phase transformation, i.e. $e^{-i\phi\hat{\sigma}_z}\mathcal{L}[\rho_S]e^{i\phi\hat{\sigma}_z}=\mathcal{L}[e^{-i\phi\hat{\sigma}_z}\rho_S e^{i\phi\hat{\sigma}_z}]$, which corresponds to Eq.~\eqref{eqn:symmetry} in the case of a single fermionic mode. Therefore, the symmetry group generated by $\mathcal{N}$ is isomorphic to $U(1)$. A complete characterization of the single-qubit phase-covariant master equation can be found in the supplementary material of Ref.~\cite{Smirne2016}.
}

\subsection{Consequences of the symmetry}
\label{sec:consequences}
{In this section we discuss a list of interesting consequences of the symmetry 
presented in Proposition 1. We start with a simple corollary:}
\begin{lemma}[Parity]
If the conditions for Proposition 1 hold, then the parity superator 
$\mathcal{P}$ is a symmetry on the superoperator level as well: 
$[\mathcal{P},\mathcal{L}]=0$.
\end{lemma}
\begin{proof}
If the conditions for Proposition 1 hold, then $[\mathcal{N},\mathcal{L}]=0$. 
But according to Eq.~\eqref{eqn:paritySupOp} 
$\mathcal{P}=\exp(i\pi\mathcal{N})$, 
thus the parity superoperator must commute with the Liouvillian as well 
{demonstrating} the assertion.
\end{proof}

Notice that the symmetries {in Proposition 1 and Corollary 1 are}, in general, 
only on the superoperator level. Indeed, we are not imposing any further 
condition on the form of the interaction and on the {spectral density} of 
environment, that is to say, the result of Eq.~\eqref{eqn:symmetry} is 
{an interesting consequence of the partial secular approximation} {only}. This 
includes cases in which the parity of the number of particles (on the operator 
level) 
is modified by the interaction with the environment. For instance, the very 
common decay of a single mode of the electromagnetic field, described as 
$\dot{\rho}=
a \rho a^\dagger-1/2 \{a^\dagger a,\rho\}$ \cite{BreuerPetruccione}, clearly 
does not conserve either $\hat{N}$ or $\hat{P}$, while {as it holds the} partial 
secular 
approximation it fulfils Eq.~\eqref{eqn:symmetry}.

How can we exploit Eq.~\eqref{eqn:symmetry} for the analysis of the open system? 
As already mentioned in Sec.~\ref{sec:symmetries}, the symmetry generated by the 
number superoperator allows us to block-diagonalize the Liouvillian in a way 
that is particularly convenient for the solution of the master equation. Indeed, 
the eigenvectors of $\mathcal{N}$ {in the representation expressed by Eq.~\eqref{eqn:superoperatorN}} are given by the tensor product of the 
diagonal basis of $\hat{H}_S$ with itself. That is to say, if we rewrite the 
system Hamiltonian as $\hat{H}_S=\sum_n \epsilon_n \ket{e_n}\bra{e_n}$, we 
choose the basis of the space of superoperators 
$\{\ket{e_n}\otimes\ket{e_m}\}_{n,m}$. This is exactly the basis we work with 
when deriving the Bloch-Redfield master equation, since it is the basis in which 
we write the jump operators \cite{BreuerPetruccione,Weiss,Cattaneo2019}. 
Therefore, if we express $\mathcal{L}$ as a matrix in the basis 
$\ket{e_n}\otimes\ket{e_m}$, and we regroup all the elements of the basis which 
are eigenvectors of $\mathcal{N}$ with the same 
eigenvalue $d$, we naturally find the blocks of the Liouvillian in such basis. 
Note that $d$ is the difference between the number of particles in the state 
$\ket{e_n}$ and the number of particles in $\ket{e_m}$. We can express this fact 
in the following proposition:
\begin{propx}[Blocks]
In a system of $M$ bosonic or fermionic modes in which Proposition 1 
holds, the Liouvillian superoperator can be divided {into} blocks as 
$\mathcal{L}=\bigoplus_{d} \mathcal{L}_d$, where $\mathcal{L}_d$ is the 
block labeled by the eigenvalue $d$ of $\mathcal{N}$. Let us write 
$\mathcal{L}_d$ as a matrix in a basis $\{\ket{e_j}\otimes\ket{e_k'}\}_{j,k}$ 
which spans its space, where $\ket{e_j}$ and $\ket{e_k'}$ are eigenvectors of 
$\hat{H}_S$. Then, if we write $\mathcal{L}_{-d}$ as a matrix in the basis 
$\{\ket{e_k'}\otimes\ket{e_j}\}_{j,k}$ these matrices satisfy 
$\mathcal{L}_d=\mathcal{L}_{-d}^*$.
\end{propx}
\begin{proof}
The Liouvillian can be block-diagonalized thanks to the symmetry expressed by 
Eq.~\eqref{eqn:symmetry}, generated by $\mathcal{N}$ whose eigenvalues label the 
blocks. 
$\mathcal{L}$ describes the dynamics of the density matrix of the system 
$\rho_S$ as in Eq.~\eqref{eqn:masterEqLiouvillian}, 
but since $(\rho_S)_{jk}=(\rho_S)_{kj}^*$, we have in the chosen 
{bases}  $\mathcal{L}_d=\mathcal{L}_{-d}^*$.
\end{proof}

Proposition 2 tells us that the symmetry in Eq.~\eqref{eqn:symmetry} not only 
provides a block division {for the Liouvillian}, but also reduces the number 
of
independent blocks, e.g. for fermions {from $2M+1$ to $M+1$}. This may greatly simplify the solution 
of the master equation, which now would live in spaces of lower dimension.
We will show in Sec.~\ref{sec:examples} some examples of this block 
diagonalization and complexity reduction.

Each block of the Liouvillian superoperator may give us important insight about 
a certain physical phenomenon of interest. 
If we know that a given block contains all the relevant information about such 
phenomenon, we may indeed analyze only this block and neglect all the rest, 
thus working in a far smaller space than the one in which $\mathcal{L}$ lives. 
{This happens, for instance, in Ref.~\cite{Bellomo2017}, where two independent 
blocks of 
the Liouvillian superoperator describing the decay of two spins (corresponding 
to the blocks discussed in Proposition 2) contain all the information about two 
different 
physical phenomena, namely superradiance and quantum synchronization. {Besides, note that all the populations of the state of the system belong to the block $\mathcal{L}_0$.

Finding the unique steady state of a relaxing Lindblad dynamics is another 
example of {the advantages entailed by the block structure of 
$\mathcal{L}$}:}
a steady state of the open dynamics is a state $\rho_{ss}$ such that 
$\mathcal{L}[\rho_{ss}]=0$. 
It always exists at least one steady state for finite systems 
\cite{Rivas2012,Baumgartner2008} and, if it is unique, then the semigroup is 
relaxing, i.e. any state is 
driven toward $\rho_{ss}$ for $t\rightarrow \infty$, and no oscillating 
coherence survives.

The unique steady state ``lives'' in the subspace of the block $\mathcal{L}_0$ 
only. {Indeed, let us call $\Pi_0$ the  projector over the eigenspace of 
$\mathcal{N}$ associated to the eigenvalue $0$. 
Then, the 
following proposition  holds: }
\begin{propx}[Steady state]
If the conditions for Proposition 1 hold and the semigroup generated by 
$\mathcal{L}$ is relaxing toward a unique steady state $\rho_{ss}$, then 
$\Pi_0[\rho_{ss}]=\rho_{ss}$. i.e. the only non-zero elements of the density 
matrix representing $\rho_{ss}$ in the excitation basis are the ones with equal 
number of excitations in the ket and in the bra. 
\end{propx}
\begin{proof}
{Let us suppose that the steady state 
$\rho_{ss}$ has a non-zero component in a subspace projected by $\Pi_d$ with $d\neq 0$: $\Pi_d[\rho_{ss}]\neq 0$. Coming back to the space of density matrices, this means that the density matrix of the steady state in the excitation basis has some non-zero elements
with different number of particles in the bra and in the
ket.
Therefore,} there exists a block $\mathcal{L}_d$ with $d\neq 0$ having a 
zero eigenvalue. Furthermore, the block $\mathcal{L}_0$ must have a zero 
eigenvalue as well, since for $\rho_{ss}$ to be a physical state it must possess 
diagonal elements. We now build a new state $\rho_{ss}'$ such that 
{$\Pi_0[\rho_{ss}']=\Pi_0[\rho_{ss}]$} and {$\Pi_k[\rho_{ss}']=0$} for all $k\neq 
0$. $\rho_{ss}'$ is a physical state (since we have obtained it by removing 
coherences from $\rho_{ss}$) and is a steady state as well, since it has 
the same elements of $\rho_{ss}$ in the space projected by {$\Pi_0$} whose 
evolution must be independent from the one of the elements in the space 
projected by  {$\Pi_d$}. Therefore, the steady state is not unique anymore 
and we have proven the assertion by contradiction.
\end{proof}

{Proposition 3 implies the corollary } that $\mathcal{L}_0$ is the only block having an eigenvalue equal to zero, while all 
the eigenvalues of the remaining blocks have negative real part. 
Another {immediate} consequence is the following.

\begin{lemma}[Steady-state coherences]
If the conditions for Proposition 1 hold and the semigroup generated by $\mathcal{L}$ is relaxing toward a unique steady state, 
then the only non-zero steady-state coherences in the excitation basis must have the same number of excitations in the ket and in the bra.
\end{lemma}

Proposition 3 is telling us that, when the semigroup dynamics is relaxing toward 
a unique steady state as it is often the case, we only need to find the 
eigenvalues and eigenvectors of the block $\mathcal{L}_0$ to characterize the 
stationary state. In particular, this restricts the range of possible 
steady-state coherences in which we may be interested, e.g. for thermodynamics 
tasks. Proposition 3 does not give information about scenarios with a broader 
space of steady states, such as in the presence of decoherence free subspaces 
and/or oscillating coherences. Further studies are needed toward this 
direction.

Finally, let us comment that Proposition 1 and Eq.~\eqref{eqn:symmetry} may not be
valid if we choose the local approach to derive the master equation 
\cite{Gonzalez2017,Hofer2017,Cattaneo2019} of a system {composed} of interacting subsystems (that can be rewritten as non-interacting 
normal modes). Indeed, the local basis used to find the jump operators would not 
coincide anymore with the diagonal basis of the normal modes of $\hat{H}_S$ 
\cite{Trushechkin2016,Cattaneo2019}, and this may create extra-terms in the 
Liouvillian superoperator which would not respect the rules discussed in 
Sec.~\ref{sec:symRes}. For some particular cases, this fact may turn the global 
approach computationally more convenient than the local one. 

\section{Examples}
\label{sec:examples}
In this section we will propose a couple of physical examples (one for fermions, 
one for bosons) 
in which the symmetry of Eq.~\eqref{eqn:symmetry} appears, and we will show how 
it significantly 
reduces the complexity of the master equation by a block diagonalization of the 
Liouvillian 
superoperator. {We choose as examples some simple low-dimensional cases, whose solution is in general already known, in order to show how to identify and employ the symmetry also in familiar scenarios. Of course, more cumbersome situations would exhibit an even more drastic dimensionality reduction.} For simplicity, from now on we will drop the hat sign over the 
operators living in the 
Liouville space $\mathsf{L}$.

\subsection{Fermions}
\label{sec:fermions}
Consider a system of $M$ non-interacting fermions, with Hamiltonian:
\begin{equation}
\label{eqn:freeHamFermions}
\sum_{k=1}^M E_k f_k^\dagger f_k.
\end{equation}
If we let the fermions interact with local and/or collective baths through an 
interaction Hamiltonian $H_I$ which satisfies {one of the conditions} 
discussed in Sec.~\ref{sec:requirements}, the Liouvillian superoperator will be 
block-diagonal with each block labeled by the eigenvalues of the operator 
$\mathcal{N}$ in Eq.~\eqref{eqn:supOpAbstr}. We will provide the explicit form 
of such a Liouvillian {for} $M=2$ fermions in 
Sec.~\ref{sec:intSpins}, {being this case of utmost importance in different 
fields 
such as quantum computation or quantum thermodynamics}. {Before that, let us 
establish} 
the dimension of each block {for any $M$}. 
Let us term $d$ an (integer) eigenvalue of $\mathcal{N}$
{assuming} values $d=-M,\ldots,-1,0,1,\ldots,M$. The dimension of the block 
$\mathcal{L}_d$ is given by the number of {excitation-basis} vectors, written in 
the tensor notation of Eq.~\eqref{eqn:isomorphism}, which have a difference 
between the number of excitations on the left and on the right of the tensor product equal to $d$. 
Taking into account all the possible combinations of suitable excitations in the 
vectors and all their possible permutations, the 
dimension of $\mathcal{L}_d$ reads:
\begin{equation}
\label{eqn:numFermions}
dim(\mathcal{L}_d)=\sum_{k=\abs{d}}^M \binom{M}{k}\cdot \binom{M}{k-\abs{d}}.
\end{equation}

\subsubsection{Two interacting spins as decoupled fermions}
\label{sec:intSpins}
Consider a system of two interacting spins with Hamiltonian:
\begin{equation}
\label{eqn:HamiltonianSpins}
H_S=\frac{\omega_1}{2}\sigma_1^z+\frac{\omega_2}{2}\sigma_2^z+\lambda 
\sigma_1^x\sigma_2^x.
\end{equation}
By employing the Jordan-Wigner transformations, a rotation and a Bogoliubov 
transformation 
{(see the discussion in Appendices~\ref{sec:JordanWigner} and~\ref{sec:coupledSpinsFree}),} we can rewrite the 
system Hamiltonian as:
\begin{equation}
\label{eqn:HamFermions}
H_S=E_1\left(2 f_1^\dagger f_1-1\right)+E_2\left(2 f_2^\dagger f_2-1\right),
\end{equation}
where $f_1$ and $f_2$ are fermionic operators satysfing the fermionic 
anticommutation rules: 
$\{f_j,f_k^\dagger\}=\delta_{jk}$, while the expressions of the energies $E_1$ 
and $E_2$ can be {found} in
Eq.~\eqref{eqn:E1E2}. The interaction eigenbasis of $H_S$ is 
$\{\ket{00}_f,\ket{01}_f,\ket{10}_f,\ket{11}_f\}$, and its relation with the 
canonical spin basis can be found in Eqs.~\eqref{eqn:groundState} 
and~\eqref{eqn:remainingStates}. {Although this transformation may appear redundant in the simple case of two qubits, it is fundamental to diagonalize more complex chains of interacting spins \cite{coleman2015introduction}, see for instance Appendix~\ref{sec:JordanWigner}.}

We couple each qubit to a separate thermal bath, such that the Hamiltonian of 
the environment is $H_E=\sum_k \Omega_k a_k^\dagger a_k+\sum_l \Omega'_l 
b_l^\dagger b_l$ and the interaction Hamiltonian reads:
\begin{equation}
\label{eqn:intHamSpins}
H_I=\sum_k g_k \sigma_1^x (a_k^\dagger+a_k)+\sum_l g'_l 
\sigma_2^x(b_l^\dagger+b_l),
\end{equation}
where $g_k$ and $g'_l$ determine the spectral densities of the baths 
\cite{BreuerPetruccione}. {As mentioned before, these} are not relevant for 
the present discussion
{and are assumed to display fast decaying correlation functions}, 
inducing a Markovian evolution. {We assume that the both baths are in a thermal state with temperature respectively $T_1$ and $T_2$. Such as a system is of fundamental importance e.g. for the understanding of quantum heat transport in quantum thermodynamics \cite{Hofer2017}.}

\begin{widetext}
Using Eqs.~\eqref{eqn:sigma1xJW} and~\eqref{eqn:sigma2xJW} we can rewrite the 
interaction Hamiltonian as:
\begin{equation}
\label{eqn:intHamFermions}
\begin{split}
H_I=&\sum_k g_k 
\left(\cos(\theta+\phi)(f_1^\dagger+f_1)+\sin(\theta+\phi)(f_2^\dagger+f_2)\right)(a_k^\dagger+a_k)\\
&+\sum_l g'_l 
\left(\cos(\theta-\phi)P(f_2^\dagger-f_2)+\sin(\theta-\phi)P(f_1^\dagger-f_1)\right)(b_l^\dagger+b_l),
\end{split}
\end{equation}
where $P$ is the parity operator, and we notice that each separate bath plays 
now the role of a common bath between the two fermionic modes. Note that 
Eq.~\eqref{eqn:intHamFermions} satisfies the second condition on the interaction 
Hamiltonian presented in Sec.~\ref{sec:requirements}.

The interacting Hamiltonian Eq.~\eqref{eqn:intHamFermions} leads to the 
following master equation:
\begin{equation}
\label{eqn:masterEqFermions}
\begin{split}
\frac{d}{dt}\rho_S(t)=&-i[H_S+H_{LS},\rho_S(t)]+\sum_{i,j=1,2} 
\gamma_{ij}^\downarrow\left(f_i\rho_S(t) f_j^\dagger-\frac{1}{2}\{f_j^\dagger 
f_i,\rho_S(t)\}\right)\\
&+\sum_{i,j=1,2} 
\gamma_{ij}^\uparrow\left(f_i^\dagger\rho_S(t) f_j-\frac{1}{2}\{f_j 
f_i^\dagger,\rho_S(t)\}\right)+\sum_{i,j=1,2} \eta^\downarrow_{ij}\left(Pf_i\rho_S(t)f_j^\dagger 
P-\frac{1}{2}\{f_j^\dagger f_i,\rho_S(t)\}\right)\\
&+\sum_{i,j=1,2} 
\eta_{ij}^\uparrow\left(f_i^\dagger P\rho_S(t) Pf_j-\frac{1}{2}\{f_j 
f_i^\dagger,\rho_S(t)\}\right),
\end{split}
\end{equation}
where the Lamb-shift Hamiltonian reads $H_{LS}=\sum_{i,j=1,2} (s_{ij}^\downarrow
f_j^\dagger f_i+s_{ij}^\uparrow
f_j f_i^\dagger)$. The coefficients
$\gamma_{ij}^\downarrow$, $\gamma_{ij}^\uparrow$, $\eta^\downarrow_{ij}$, $\eta^\uparrow_{ij}$, $s_{ij}^\downarrow$ and $s_{ij}^\uparrow$ depend on the spectral densities of 
the baths, {on the temperature} and on the weights of each term in the interaction Hamiltonian. We do 
not provide their explicit value here, and we refer the interested reader to 
the derivation in Refs.~\cite{BreuerPetruccione,Cattaneo2019}.

We now find the Liouvillian superoperator representing the master 
equation~\eqref{eqn:masterEqFermions} in the tensor product notation, as in 
Eq.~\eqref{eqn:LiouvillianBraket}. We identify five symmetry blocks of 
$\mathcal{L}$, associated to the following 
{bases}: 
$\ket{11}_f\otimes\ket{11}_f,\ket{10}_f\otimes\ket{10}_f,\ket{10}_f\otimes\ket{
01}_f,\ket{01}_f\otimes\ket{10}_f,\ket{01}_f\otimes\ket{01}_f,\ket{00}
_f\otimes\ket{00}_f$ corresponding to $\mathcal{N}=0$; 
$\ket{11}_f\otimes\ket{10}_f,\ket{11}_f\otimes\ket{01}_f,\ket{10}_f\otimes\ket{
00}_f,\ket{01}_f\otimes\ket{00}_f$ corresponding to $\mathcal{N}=1$; 
$\ket{10}_f\otimes\ket{11}_f,\ket{01}_f\otimes\ket{11}_f,\ket{00}_f\otimes\ket{
10}_f,\ket{00}_f\otimes\ket{01}_f$ corresponding to $\mathcal{N}=-1$; 
$\ket{11}_f\otimes\ket{00}_f$ corresponding to $\mathcal{N}=2$; 
$\ket{00}_f\otimes\ket{11}_f$ corresponding to $\mathcal{N}=-2$. The Liouvillian 
can be written as $\mathcal{L}=\bigoplus_{d=-2}^2 \mathcal{L}_d$, where the 
matrices representing each block in the associated basis are:
\begin{equation}
\label{eqn:mu0}
\mathcal{L}_0=\begin{pmatrix}
-\gamma_0^\downarrow-\eta_0^\downarrow&\gamma_{22}^\uparrow+\eta_{22}^\uparrow&-\gamma_{21}^\uparrow-\eta_{21}^\uparrow&-\gamma_{12}^\uparrow-\eta_{12}^\uparrow&\gamma_{11}^\uparrow+\eta_{11}^\uparrow&0\\
\gamma_{22}^\downarrow+\eta_{22}^\downarrow&-\xi^\downarrow_{11}-\xi^\uparrow_{22}&is_{21}-\frac{\xi^\downarrow_{12}-\xi^\uparrow_{21}}{2}&-is_{12}-\frac{\xi^\downarrow_{21}-\xi^\uparrow_{12}}{2}&0&\gamma^\uparrow_{11}+\eta^\uparrow_{11}\\
-\eta_{21}^\downarrow-\gamma_{21}^\downarrow&is_{12}-\frac{\xi^\downarrow_{21}-\xi^\uparrow_{12}}{2}
&-i(\omega_1'-\omega_2'	)-\frac{\xi_0^\downarrow+\xi_0^\uparrow}{2}&0&-is_{12}-\frac{\xi^\downarrow_{21}-\xi^\uparrow_{12}}{2}&\gamma^\uparrow_{12}+\eta^\uparrow_{12}\\
-\eta^\downarrow_{12}-\gamma^\downarrow_{12}&-is_{21}-\frac{\xi^\downarrow_{12}-\xi^\uparrow_{21}}{2}
&0&i(\omega_1'-\omega_2'	)-\frac{\xi_0^\downarrow+\xi_0^\uparrow}{2}&is_{21}-\frac{\xi^\downarrow_{12}-\xi^\uparrow_{21}}{2}&\gamma^\uparrow_{21}+\eta^\uparrow_{21}\\
\gamma^\downarrow_{11}+\eta^\downarrow_{11}&0&-is_{21}-\frac{\xi^\downarrow_{12}-\xi^\uparrow_{21}}{2}&is_{12}-\frac{\xi^\downarrow_{21}-\xi^\uparrow_{12}}{2}&-\xi^\downarrow_{22}-\xi^\uparrow_{11}&\gamma_{22}^\uparrow+\eta^\uparrow_{22}\\
0&\gamma^\downarrow_{11}+\eta^\downarrow_{11}&\gamma^\downarrow_{12}+\eta^\downarrow_{12}&\gamma^\downarrow_{21}+\eta^\downarrow_{21}
&\gamma_{22}^\downarrow+\eta^\downarrow_{22}&-\gamma_0^\uparrow-\eta_0^\uparrow
\end{pmatrix},
\end{equation}
\begin{equation}
\label{eqn:mu1}
\mathcal{L}_1=\begin{pmatrix}
-i\omega_2'-\xi^\downarrow_{11}-\frac{\xi^\downarrow_{22}+\xi^\uparrow_{22}}{2}&is_{21}-\frac{\xi^\downarrow_{12}-\xi^\uparrow_{21}}{2}&\eta^\uparrow_{21}-\gamma^\uparrow_{21}&\gamma^\uparrow_{11}-\eta^\uparrow_{11}\\
is_{12}-\frac{\xi^\downarrow_{21}-\xi^\uparrow_{12}}{2}&-i\omega_1'-\xi^\downarrow_{22}-\frac{\xi^\downarrow_{11}+\xi^\uparrow_{11}}{2}&\eta^\uparrow_{22}-\gamma^\uparrow_{22}&\gamma^\uparrow_{12}-\eta^\uparrow_{12}\\
\eta^\downarrow_{21}-\gamma^\downarrow_{21}&\eta^\downarrow_{22}-\gamma^\downarrow_{22}&-i\omega_1'-\xi^\uparrow_{22}-\frac{\xi^\downarrow_{11}+\xi^\uparrow_{11}}{2}&-is_{12}-\frac{\xi^\downarrow_{21}-\xi^\uparrow_{12}}{2}\\
\gamma^\downarrow_{11}-\eta^\downarrow_{11}&\gamma^\downarrow_{12}-\eta^\downarrow_{12}&-is_{21}-\frac{\xi^\downarrow_{12}-\xi^\uparrow_{21}}{2}&-i\omega_2'-\xi^\uparrow_{11}-\frac{\xi^\downarrow_{22}+\xi^\uparrow_{22}}{2}\\
\end{pmatrix},
\end{equation}
$\mathcal{L}_{-1}=\mathcal{L}_1^*$, 
$\mathcal{L}_2=-i(\omega_1'+\omega_2')-\xi^\downarrow_0-\xi^\uparrow_0$ and 
$\mathcal{L}_{-2}=\mathcal{L}_2^*$. When convenient, we have used the abbreviations $\omega_1'=2E_1+s^\downarrow_{11}-s^\uparrow_{11}$, 
$\omega_2'=2E_2+s^\downarrow_{22}-s^\uparrow_{22}$, $s_{ij}=s_{ij}^\downarrow-s_{ji}^\uparrow$, $\gamma_{0}^{\downarrow\uparrow}=\gamma_{11}^{\downarrow\uparrow}+\gamma_{22}^{\downarrow\uparrow}$, $\eta_{0}^{\downarrow\uparrow}=\eta_{11}^{\downarrow\uparrow}+\eta_{22}^{\downarrow\uparrow}$, $\xi_{ij}^{\downarrow\uparrow}=\gamma_{ij}^{\downarrow\uparrow}+\eta_{ij}^{\downarrow\uparrow}$ and $\xi_{0}^{\downarrow\uparrow}=\gamma_{0}^{\downarrow\uparrow}+\eta_{0}^{\downarrow\uparrow}$.
\end{widetext}

\begin{figure}[t!]
\centering
\includegraphics[scale=0.17]{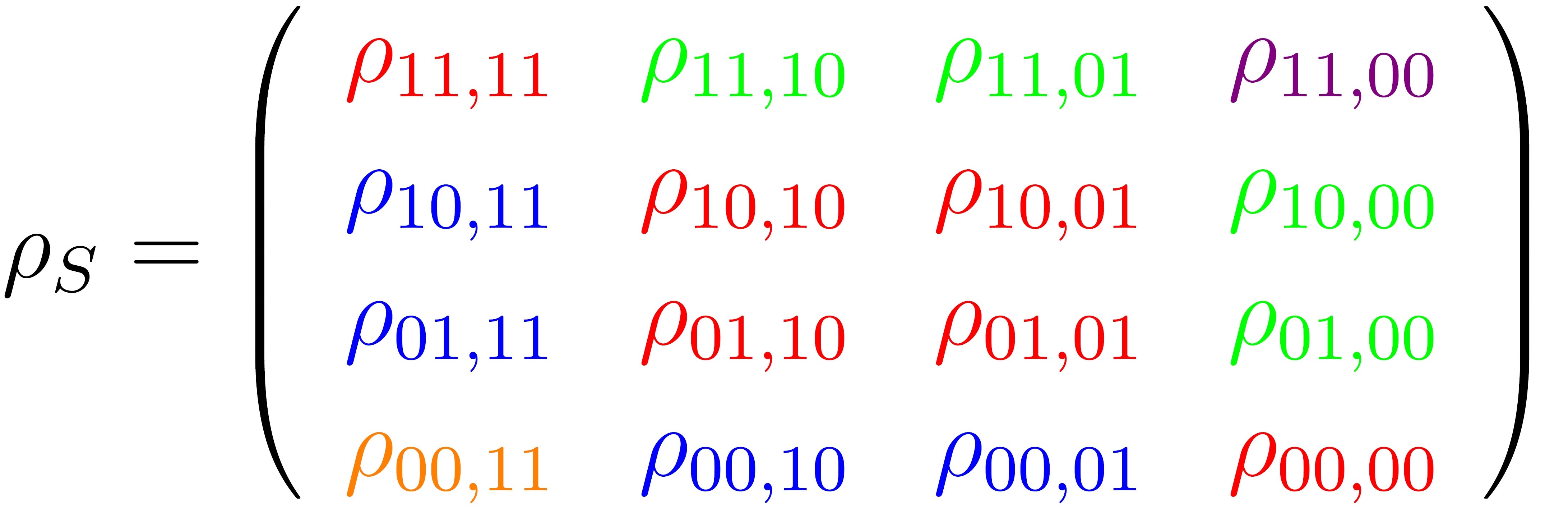}\\
\caption{(Color online) Density matrix of the state of the system with 
Hamiltonian equation~\eqref{eqn:HamFermions} in the fermionic interactions 
basis. The master equation driven by the Liouvillian $\mathcal{L}$ couples only 
elements of the density matrix with the same color. In particular, the block 
$\mathcal{L}_0$ is represented by the color red, $\mathcal{L}_1$ by the color 
green, $\mathcal{L}_{-1}$ by the color blue, $\mathcal{L}_{2}$ by the color 
violet and $\mathcal{L}_{-2}$ by the color orange.}
\label{fig:state}
\end{figure}

{Note that {assuming} a local master equation instead of 
Eq.~\eqref{eqn:masterEqFermions} would 
{lead to} extra-terms connecting, for instance, the block $\mathcal{L}_0$ 
with the blocks 
$\mathcal{L}_{\pm 2}$ \cite{Cattaneo2019}. Therefore, the block decomposition 
would not be valid in 
this case.}

A very similar structure was found for the Liouvillian of two uncoupled spins in 
a common bath \cite{Bellomo2017}, where the block separation was exploited to 
find the analytical eigenvalues describing the decay of the system. We thus 
understand the help brought by the symmetry in Eq.~\eqref{eqn:symmetry} to the 
present example: instead of having to find the eigenvalues and eigenvectors of a 
$16\times 16$ matrix, we restrict {ourselves} to the analysis of a $6\times 6$ 
and a $4\times 4$ matrix. Fig.~\ref{fig:state} depicts how the elements of the 
density matrix of the system written in the excitation basis appear in separate 
blocks of the master equation (each color representing an independent block). 

{Furthermore, if we are interested in finding the steady state of the evolution and the latter is unique, 
we just have to analyze the matrix $\mathcal{L}_0$. To be sure that the condition on the uniqueness holds, one has to check that no decoherence-free subspaces are present. Their appearance can be detected a priori using different conditions on the interaction Hamiltonian or on the master equation \cite{Lidar2003}, otherwise they can be revealed by the presence of more than one null eigenvalue in the spectrum of the Liouvillian superoperator. Since we have set non-zero, unbalanced temperatures of the baths, the steady state may contain coherences as well, 
but only the ones corresponding to the eigenvalue $0$ of $\mathcal{N}$, namely 
$\rho_{10,01}$ and $\rho_{01,10}$. The same steady-state coherences were found using a non-secular master equation in a couple of recent works \cite{PhysRevA.99.042320,Huangfu2018}. We will provide an example of the appearance of these coherences below and in the next example about harmonic oscillators.} 

Note that, if we had 
performed the full secular approximation instead of the partial one, we would 
have introduced a broader symmetry on the superoperator level, dividing the 
block $\mathcal{L}_0$ into two additional parts. Indeed, if the spectrum of 
$H_S$ is non-degenerate, the full secular approximation decouples coherences and 
populations \cite{BreuerPetruccione}. {The symmetry generated by $\mathcal{N}$ is therefore providing us with new ``selection rules'' that indicate the allowed transitions between elements of the density matrix: in the partial secular regime some of the coherences may exchange ``amplitude'' with the diagonal elements. We can visualize this through a concrete case of the two-coupled-qubits example: let us consider a scenario with $\omega_1=1$, $\omega_2=1$, $\lambda=0.01$, $T_1=\omega_1/k_B$, $T_2=\omega_1/10k_B$ and Ohmic spectral densities (from now on for simplicity we use dimensionless units for time and energy). Using these values, the fermionic energies read $2E_1=1.01005$ and $2E_2=0.99005$. $\mu$ denotes the strength of the qubit-bath coupling, and considering the weak-coupling limit we set $\mu=10^{-1.5}$. According to Eq.~\eqref{eqn:tauR}, $\tau_R\approx1000$ and therefore the partial secular approximation must conserve the terms in the master equation associated to the frequency difference $2E_1-2E_2=0.02$, which does not satisfy the condition in Eq.~\eqref{eqn:partialSecular}. Using the above values, we can calculate the coefficients of the master equation~\eqref{eqn:masterEqFermions} according to the discussion in Ref.~\cite{BreuerPetruccione} and we can compute the dynamics of the two qubits. According to the suitable conditions \cite{Lidar2003}, we have checked that no decoherence-free subspace is present in this scenario.

\begin{figure}[t]
\includegraphics[scale=.8]{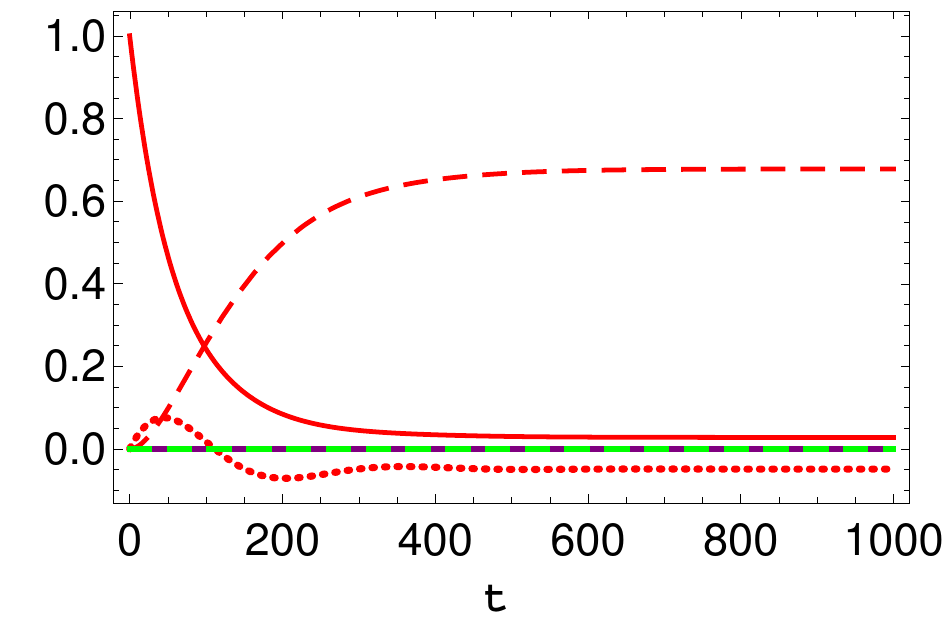}
\caption{(Color online) Mean value of different two-qubit observables as a function of time, when the evolution starts in the state $\ket{11}_f$ and we use the master equation~\eqref{eqn:masterEqFermions} in partial secular approximation. As defined in the main text, we plot $P_{11}(t)$ (solid red), $P_{00}(t)$ (dashed red), $C_{0}(t)$ (dotted red), $C_{1}(t)$ (dashed green), $C_{2}(t)$ (solid purple). $C_1$ and $C_2$ remain null during the evolution, while using the master equation in partial secular approximation $C_0$ varies and stabilizes to a non-zero value after the thermalization time, showing the appearance of steady-state coherences. On the contrary, the full secular approximation would keep $C_0$ null as well.}
\label{fig:ev}
\end{figure}

We now want to visualize how the partial secular approximation induces selection rules between different elements of the density matrix: we consider the evolution of the two qubits starting from the excited fermionic state $\ket{11}_f$, i.e. $\rho_S(0)=\ket{11}_f\!\bra{11}$, whose dynamics is driven by the block $\mathcal{L}_0$. We monitor the expectation value of some observables which pertain to different blocks of the Liouvillian superoperator as a function of time, in particular we choose: 
\begin{itemize}
\item $P_{11}(t)=\Tr[\rho_S(t) \ket{11}_f\!\bra{11}]=\rho_{11,11}(t)$,
\item $P_{00}(t)=\Tr[\rho_S(t) \ket{00}_f\!\bra{00}]=\rho_{00,00}(t)$,
\item $\begin{aligned}[t]
C_0(t)=&\Tr[\rho_S(t) (f_1^\dagger f_2+f_2^\dagger f_1)]\\
=&2\Re[\rho_{01,10}(t)],
\end{aligned}$
\item $\begin{aligned}[t]
C_1(t)=&\Tr[\rho_S(t) (f_1+f_1^\dagger)]\\
=&2\Re[\rho_{11,01}(t)]+2\Re[\rho_{10,00}(t)],
\end{aligned}$
\item $\begin{aligned}[t]
C_2(t)=&\Tr[\rho_S(t) (f_1^\dagger f_2^\dagger+f_2f_1)]\\
=&2\Re[\rho_{11,00}(t)].
\end{aligned}$
\end{itemize}
$P_{11}$, $P_{00}$ and $C_0$ depend on elements of the block $\mathcal{L}_0$, while $C_1$ of the block $\mathcal{L}_1$ and $C_2$ of the block $\mathcal{L}_2$. Figure~\ref{fig:ev} depicts their evolution: $P_{11}$ is the only non-zero mean value at time $0$. Therefore, the symmetry brought by the partial secular approximation $[\mathcal{N},\mathcal{L}]=0$ denies the possibility that $C_1$ and $C_2$ may change their value during the evolution, since their dynamics is driven by blocks different from $\mathcal{L}_0$. On the contrary, $C_0$ increases and stabilizes to a non-zero value at infinite time, since it is the mean value of the observable expressing the exchange of excitations between the fermions, whose dynamics is driven by $\mathcal{L}_0$.

Let us now briefly discuss how things would change if we applied the full secular approximation to derive the master equation~\eqref{eqn:masterEqFermions}, i.e. if we removed the terms with $i\neq j$. The full secular approximation decouples coherences and populations \cite{BreuerPetruccione}, therefore, in the scenario discussed before, not only it would inhibit any transition that may ``activate'' $C_1$ and $C_2$, but would also keep $C_0$ null, given that the latter depends on the density matrix element $\rho_{01,10}$. This means that, in Fig.~\ref{fig:ev}, the full secular approximation would make the dotted red line ($C_0$) overlap with the green and purple lines ($C_1$ and $C_2$), thus proving itself not suitable to treat the current scenario \cite{Cattaneo2019}.

}

\subsection{Bosons}
\label{sec:bosons}
If we consider a generic bosonic system for which Eq.~\eqref{eqn:symmetry} 
holds, we will still have a block diagonalization of the Liouvillian 
superoperator which will simplify the resolution of the master equation, but 
each block will have infinite dimension. Here, we want to focus on a simpler 
case in which the symmetry expressed by Eq.~\eqref{eqn:symmetry} leads to a 
dimensionality reduction as well: we restrict ourselves to the space of Gaussian 
states \cite{Ferraro2005} and we consider only a master equation conserving 
Gaussianity. Therefore, we {only need to analyze} the dynamics of the covariance 
matrix, neglecting any displacement which may be eliminated through a suitable 
transformation.

Let us consider a system of $M$ non-interacting bosons with Hamiltonian:
\begin{equation}
\label{eqn:HamBosons}
H_S=\sum_{k=1}^M E_k a_k^\dagger a_k.
\end{equation}
Given the presence of local or common baths leading to a Gaussian Markovian 
master equation, we want to study the dynamics of a Gaussian state with no 
displacement. For convenience, we choose to write the covariance matrix of the 
state using the creation and annihilation operators, i.e. a generic element of 
the covariance matrix may be written in one of these three forms:
\begin{equation}
\label{eqn:covMatrixDef}
\langle a_i^\dagger a_j^\dagger\rangle \quad\textnormal{  or  } \quad \langle 
a_i^\dagger a_j \rangle \quad \textnormal{ or }\quad\langle a_i a_j \rangle,
\end{equation}
where the average is performed on the chosen Gaussian state. We define as 
$\delta$ the difference between number of creations and number of 
annihilations in an element of the covariance matrix 
Eq.~\eqref{eqn:covMatrixDef}, assuming values 
$2,0,-2$, respectively.

It is easy to understand what the symmetry defined by Eq.~\eqref{eqn:symmetry} 
is telling us about the evolution of the covariance matrix: the dynamics of an 
element of the covariance matrix with value {$\delta$} can only be a function of 
elements of the covariance matrix with the same value {$\delta$}\footnote{{This property can be extended}
to non-Gaussian states, where the $n$-th moments 
must be taken into account: the master equation describing the evolution of the 
$n$-th moment $\langle\underbrace{ a_i^\dagger a_j^\dagger\ldots}_\text{l 
creations}\underbrace{\ldots a_r a_s}_\text{n-l annihilations}\rangle$ with 
{$\delta$}$=2l-n$ can only be a function of $m$-th moments with the same value 
of {$\delta$} (difference between number of creations and number of 
annihilations in the moment).}. We can collect the elements of the covariance 
matrix Eq.~\eqref{eqn:covMatrixDef} (which cannot be trivially obtained through 
commutations of the other elements) in a vector $\mathbf{x}$. The evolution of 
$\mathbf{x}$ 
as a function of time is then given by the formula
\begin{equation}
\label{eqn:evCovMatrix}
\frac{d\mathbf{x}}{dt}=B \mathbf{x}+\mathbf{b}.
\end{equation}
The matrix $B$ is block-diagonalized 
labelling each block with the value {$\delta$}: {$B=\bigoplus_{\delta=-2,0,2} 
B_\delta$}. Furthermore, $\langle a_i^\dagger a_j^\dagger \rangle=\langle a_i 
a_j\rangle^*$ and the symmetry assures us that these two moments do not couple 
in the master equation, therefore the block $B_{2}$ is trivially obtained by the 
block $B_{-2}$. We can now calculate the dimension of each block {$B_{\delta}$}:
\begin{equation}
\label{eqn:dimBr}
\begin{split}
dim(B_0)&=M^2, \\
dim(B_{\pm 2})&=\binom{M+2-1}{2}.
\end{split}
\end{equation}

\subsubsection{Two bosons in a common bath}
{As an example, we} consider the system of two displaced non-interacting bosons 
with Hamiltonian:
\begin{equation}
\label{eqn:displacedH}
H_S=\sum_{k=1,2} (\omega_k a_k^\dagger a_k -\alpha_k a_k-\alpha_k^* a_k^\dagger 
).
\end{equation}
The Hamiltonian can be recast in the standard form of 
Eq.~\eqref{eqn:HamiltonianNonIntModes} through a suitable displacement operator 
$D(\boldsymbol{\alpha})$ \cite{Ferraro2005}: $D(\boldsymbol{\alpha})^\dagger a_k 
D(\boldsymbol{\alpha})=a_k+\alpha_k$. Therefore we have:
\begin{equation}
\label{eqn:HamBosonsB}
H_S=\sum_{k=1,2} E_k a_k^\dagger a_k,
\end{equation}
which describes two non-interacting harmonic oscillators. We couple the system 
to a common bosonic environment $H_E=\sum_l \Omega_l c_l^\dagger c_l$ in a 
thermal state with temperature $T>0$. The {system-bath interaction} Hamiltonian is:
\begin{equation}
\label{eqn:intHamBosons}
H_I=\sum_l g_l(a_1+a_1^\dagger+a_2+a_2^\dagger)(c_l+c_l^\dagger),
\end{equation}
where $g_l$ determines the spectral density, which is not relevant for the 
present discussion. The evolution of the system coupled to the environment is 
given by the master equation with Liouvillian:
\begin{equation}
\label{eqn:LiouvillianBosons}
\begin{split}
\mathcal{L}^\dagger[O]=&i[H_S+H_{LS},O]\\ 
&+\sum_{ij=1,2}\gamma_{ij}^\downarrow \left(a_i^\dagger O a_j -\frac{1}{2}\{a_i^\dagger
a_j,O\}\right)\\
&+\sum_{ij=1,2}\gamma_{ij}^\uparrow \left(a_i O a_j^\dagger 
-\frac{1}{2}\{a_i a_j^\dagger,O\}\right),\\
\end{split}
\end{equation}
where $O$ is an operator acting on the Hilbert space of the system, and 
$\gamma_{ij}^\downarrow$ and $\gamma_{ij}^\uparrow$ are respectively the 
coefficients describing the decay and the absorption, which depend on the 
spectral density and on the temperature of the environment 
\cite{BreuerPetruccione}. The Lamb-shift Hamiltonian reads:
\begin{equation}
\label{eqn:lambShiftBosons}
H_{LS}=\sum_{ij=1,2} s_{ij} a_i^\dagger a_j .
\end{equation}
The elements $\gamma_{12}$ and $\gamma_{21}$ are different from zero only if the 
harmonic oscillators are slightly detuned {(or not detuned at all)} \cite{Cattaneo2019}. {We remind that,
assuming that the initial state is Gaussian, then it will remain Gaussian due to 
the form of Eq.~\eqref{eqn:LiouvillianBosons}.}

{The relevant elements of the covariance matrix can be collected in a vector
$\mathbf{x}$ of} dimension $10$. In particular, we choose to parametrize it 
according to the basis $\langle a_1^\dagger a_1^\dagger\rangle$, $\langle 
a_2^\dagger a_2^\dagger\rangle$, $\langle a_1^\dagger a_2^\dagger\rangle$ with 
{$\delta=2$}. $\langle a_1 a_1\rangle$, $\langle a_2 a_2\rangle$, $\langle a_1 
a_2 \rangle$ with {$\delta=-2$}. $\langle a_1^\dagger a_1\rangle$, $\langle 
a_2^\dagger a_2\rangle$, $\langle a_1^\dagger a_2 \rangle$, $\langle a_2^\dagger 
a_1 \rangle$ with {$\delta=0$}. The master equation describing the evolution of 
$\mathbf{x}$ has the form of Eq.~\eqref{eqn:evCovMatrix}. The vector 
$\mathbf{b}$ can be written as {$\mathbf{b}=\oplus_{\delta=-2,0,2} 
\mathbf{b}_\delta$}. We have that 
$\mathbf{b}_{\pm 2}=\boldsymbol{0}$, while
\begin{equation}
\label{eqn:vectorB}
\mathbf{b}_0=\begin{pmatrix}
\gamma^\uparrow_{11}\\
\gamma^\uparrow_{22}\\
\gamma^\uparrow_{12}\\
\gamma^\uparrow_{21}
\end{pmatrix}.
\end{equation}
\begin{widetext}
The matrix $B$ describes how the elements of the covariance matrix are coupled 
together in the master equation, and it is block-diagonal according to 
{$B=\bigoplus_{\delta=-2,0,2} B_\delta$}. The blocks are given by:
\begin{equation}
\label{eqn:block0}
B_0=\begin{pmatrix}
-\gamma^{b}_{1}&0&\frac{-2i 
s_{12}-\gamma^\downarrow_{12}+\gamma^\uparrow_{21}}{2}&\frac{2i 
s_{21}-\gamma^\downarrow_{21}+\gamma^\uparrow_{12}}{2}\\
0&-\gamma^{b}_{2}&\frac{2i 
s_{12}-\gamma^\downarrow_{12}+\gamma^\uparrow_{21}}{2}&\frac{-2i 
s_{21}-\gamma^\downarrow_{21}+\gamma^\uparrow_{12}}{2}\\
\frac{-2i s_{21}-\gamma^\downarrow_{21}+\gamma^\uparrow_{12}}{2}&\frac{2i 
s_{21}-\gamma^\downarrow_{21}+\gamma^\uparrow_{12}}{2}&i\Delta\omega-\frac{
\gamma^{b}_{1}+\gamma^{b}_{2}}{2}&0\\
\frac{2i s_{12}-\gamma^\downarrow_{12}+\gamma^\uparrow_{21}}{2}&\frac{-2i 
s_{12}-\gamma^\downarrow_{12}+\gamma^\uparrow_{21}}{2}&0&-i\Delta\omega-\frac{
\gamma^{b}_{1}+\gamma^{b}_{2}}{2}\\
\end{pmatrix},
\end{equation}
\begin{equation}
\label{eqn:block2}
B_{-2}=\begin{pmatrix}
-2i 
E'_1-\gamma^{b}_{1}&-2is_{12}-\gamma^\downarrow_{12}
+\gamma^\uparrow_{21}&0\\
-is_{21}-\frac{\gamma^\downarrow_{21}-\gamma^\uparrow_{12}}{2}&-i (E'_1+E'_2) 
-\frac{\gamma^{b}_{1}+\gamma^{b}_{2}}{2}&-is_{12}-\frac{\gamma^\downarrow_{12
}-\gamma^\uparrow_{21}}{2}\\
0&-2is_{21}-\gamma^\downarrow_{21}+\gamma^\uparrow_{12}&-2i 
E'_2-\gamma^{b}_{2}\\
\end{pmatrix},
\end{equation}
and $B_2=B_{-2}^*$. We have defined $E_1'=E_1+s_{11}$, $E_2'=E_2+s_{22}$, 
$\Delta\omega=E_1'-E_2'$, 
$\gamma_j^{b}=\gamma_{jj}^\downarrow-\gamma_{jj}^\uparrow$.
\end{widetext}
To find the steady state of the system we have to solve the equation 
$B\mathbf{x}_{ss}+\mathbf{b}=\boldsymbol{0}$. Therefore, the elements of the 
covariance matrix with {$\delta\neq 0$} vanish in the steady state. On the 
contrary, in the case in which the harmonic oscillators are slightly detuned and 
$\gamma_{12}\neq 0$, all the elements with {$\delta=0$} have a non-zero 
component for $t\rightarrow\infty$, and in particular $\langle a_1^\dagger 
a_2\rangle_{ss}$ and $\langle a_2^\dagger a_1\rangle_{ss}$ do not vanish, i.e. 
we observe steady-state coherences. The analytical form of the steady {state} 
can be obtained by solving the system of four differential equations given by 
{$B_0 \mathbf{x}_{ss}^{(\delta=0)}+\mathbf{b}_0$=0.}

\section{Discussion and conclusions}
\label{sec:conclusions}
In this paper we have shown how the Liouvillian superoperator $\mathcal{L}$ of a 
broad class of open quantum systems can be block-diagonalized through a symmetry 
on the superoperator level, namely the invariance under the action of the number 
superoperator $\mathcal{N}$, defined in Eq.~\eqref{eqn:superoperatorN}, such 
that $[\mathcal{N},\mathcal{L}]=0$. This symmetry arises when we derive the 
standard Bloch-Redfield master equation of the open system applying a suitable 
partial secular approximation whose condition is given by 
Eq.~\eqref{eqn:partialSecular}. 
{The requirements for the microscopic model are that the
system Hamiltonian can be recast as  $M$}
non-interacting bosonic or fermionic modes (Eq.~\eqref{eqn:HamiltonianNonIntModes})
and that the system operators  in the 
interaction Hamiltonian  satisfy the conditions discussed at the end of 
Sec.~\ref{sec:requirements}, which are usually fulfilled in the 
majority of physical systems of importance to quantum information or condensed 
matter physics. This includes, for instance, any system with Hamiltonian 
quadratic in the bosonic or fermionic operators, and coupled to a thermal bath 
through operators which are linear in the field operators, as well as some spin 
systems.

{The existence of the symmetry is formalized and proven in Proposition 1. 
Corollary 1 states that such symmetry implies the invariance under the action of 
the parity superoperator as well. Proposition 2 shows that we can exploit 
Proposition 1 to decompose the Liouvillian superoperator into blocks, and 
that in the fermionic case only $M+1$ of them are independent. This greatly reduces the complexity of 
the master equation. Furthermore, each block may be the only part of the 
Liouvillian we have to manipulate in order to find a certain physical quantity, 
for example Proposition 3 shows that{, when unique}, the steady state is 
determined only by one single block. This implies that the allowed steady-state 
coherences in the excitation basis of an {unique} steady state are only the 
ones with equal number of excitations in the ket and the bra, as formalized in 
Corollary 2.}

A couple of examples are also discussed. In Sec.~\ref{sec:fermions} we have 
found the dimension of each block $\mathcal{L}_d$ of the Liouvillian 
superoperator in the case of a system of $M$ fermionic modes 
(Eq.~\eqref{eqn:numFermions}), and we have shown how to apply this to a system of 
two coupled qubits. In this scenario, an originally $16\times 16$ Liouvillian is 
decomposed into five blocks of dimension $6$, $4$, $4$, $1$ and $1$. The 
information about the steady state is contained in the $6\times 6$ block only. 
The decomposition greatly simplifies the master equation and also allows to 
obtain some analytical solutions. Then, in Sec.~\ref{sec:bosons} we have 
discussed the case of bosons, focusing in particular on Gaussian states. We have 
shown how to decompose the equation for the evolution of the covariance matrix 
employing the symmetry of the number superoperator, and we have applied it to 
study the case of two harmonic 
oscillators in a common bath. In the presence of small detuning, we have 
detected steady-state coherences by focusing on a system of only 4 linear 
equations, instead of the original system of 10 equations.

{
These results may be relevant in disparate fields. For instance, reducing the complexity of the master equation describing transport in quantum systems is of great importance \cite{Dorn2019}, and our discussion may be especially relevant for materials which exhibit \textit{quasi-degeneracies} in the Hamiltonian spectrum and thus require a master equation in partial secular approximation \cite{Darau2009}. The latter is also important to study the heat current from two unbalanced reservoirs, since it solves any deficiency that a global master equation may display with respect to a local one \cite{Cattaneo2019}. The symmetry generated by $\mathcal{N}$ may be also relevant in the field of quantum metrology. Indeed, as discussed in Sec.~\ref{sec:symRes} it may be seen as a generalization of the concept of phase-covariant master equation, which plays a fundamental role in defining the limits for the frequency estimation of a single qubit \cite{Smirne2016,Haase2018,Haase2018c,Liuzzo-Scorpo2018}. Therefore, a protocol for the frequency estimation of multiple detuned qubits should rely on our result to distinguish between the possible noise models and their origin. Finally, Proposition 3 and Corollary 2 are very relevant in quantum thermodynamics and quantum thermalization. Indeed, they define a strict law on the possible steady-state coherences that may appear in the steady state of a Markovian process. The relation between coherences and diagonal elements is also important to improve the performance of quantum thermal machines \cite{Vinjanampathy2016}. 
}

{Possible extensions of this work could address other situations
where the number superoperator symmetry 
can arise.}
{In particular, beyond stationary 
environments considered here,} non-stationary autocorrelation functions 
of the bath would add a temporal dependence to the coefficients of the 
Lamb-shift Hamiltonian and {of the} dissipator in Eqs.~\eqref{eqn:lambShift} 
and~\eqref{eqn:dissipator}. This would affect the way in which we perform 
the partial secular approximation. As a consequence, there may exist scenarios 
in which the symmetry is broken. Consider for instance a single-mode 
electromagnetic field in a squeezed bath \cite{BreuerPetruccione}: the master 
equation would contain terms of the form $a\rho_S a$, where $a$ is the 
annihilation operator of the field. Clearly, in this case 
$[\mathcal{N},\mathcal{L}]\neq 0$. Note however that we would still recover the 
symmetry of the parity superoperator: $[\mathcal{P},\mathcal{L}]
=0$. Different scenarios may arise considering different states of the 
environment, and further investigation is needed to extent our work to these 
cases. 

{A further direction could be exploring non-linear scenarios
beyond quadratic system Hamiltonians, even if the latter}
include many bosonic and fermionic systems of interest to quantum information. 
{In particular, for systems of many coupled spins {where the number of spin excitations is not conserved}, even if diagonalization
through the Jordan-Wigner transformations is possible, the resulting } fermionic 
Hamiltonian generally depends on a collective phase, violating the 
``non-interacting'' condition. In these scenarios, the validity of the symmetry 
$[\mathcal{N},\mathcal{L}]=0$ must be checked case by case, using the physical 
considerations discussed in Sec.~\ref{sec:block}. On the contrary, the block 
decomposition holds {for any excitation-preserving system of spins} in common or 
separate thermal baths {(with a final requirement on the 
interaction Hamiltonian).} 

{The extension of Proposition 3 to 
scenarios with more than one steady state, e.g. in the presence of decoherence 
free subspaces or oscillating coherences, would  also be interesting. In particular, some }open questions 
not addressed here are: does the block $\mathcal{L}_0$ contain all the 
information about any steady state of the system? If not, are there particular 
cases in which this holds? Can we find an analogous theorem for oscillating 
coherences? Investigation about the same symmetry for non-Markovian master 
equations in the weak coupling limit may be interesting as well. In particular,  we expect to find the same results for the case of a time-local non-Markovian master equation in the \textit{secular regime} \cite{Haikka2010}, while non-secular terms would break the symmetry. {Finally, it 
would be useful to employ our findings to implement a fast, manageable code to 
solve the dynamics of the open system by exploiting its symmetry, as already 
done for the case of identical atoms \cite{PhysRevA.98.063815}.}

\begin{acknowledgments}
MC acknowledges 
interesting discussions with Bassano Vacchini, Bruno Bellomo, Giacomo 
Guarnieri and Gerhard Dorn, and thanks Salvatore Lorenzo for useful hints.
The authors acknowledge funding from MINECO/AEI/FEDER through projects EPheQuCS
FIS2016-78010-P, the María de Maeztu Program
for Units of Excellence in R$\&$D (MDM-2017-0711), and the
CAIB postdoctoral program, support from CSIC Research Platform PTI-001 and 
 partial funding of MC from Fondazione Angelo della Riccia.

\end{acknowledgments}

\appendix
\section{The formalism of GKLS master equations}
\subsection{From the Bloch-Redfield to the Lindblad equation}
\label{sec:blochLin}
The fact that, in general, the Bloch-Redfield master equation does not preserve 
positivity and is not in the GKLS form (or Lindblad form) 
\cite{Gorini1976a,Lindblad1976} is a very well-known issue \cite{Benatti2005}. 
The standard procedure to derive a Markovian master equation makes use of the 
full secular approximation \cite{BreuerPetruccione,Rivas2012}, i.e. removes all 
the terms with $\omega\neq\omega'$ in Eqs.~\eqref{eqn:lambShift} 
and~\eqref{eqn:dissipator}, in order to recover the semigroup structure of a 
master equation in the Lindblad form. This, however, may lead to major mistakes 
when the condition in Eq.~\eqref{eqn:partialSecular} is not fulfilled 
\cite{Cattaneo2019}. Nonetheless, some recent studies have shown that the PSA 
performed through a suitable coarse-graining does lead to a GKLS master equation 
\cite{Cresser2017a,Farina2019}{, as can be also found in a previous work which 
did not mention the PSA \cite{PhysRevA.78.022106}}. This method of applying the 
PSA is analogous to the one used in 
the present paper, based on the condition in Eq.~\eqref{eqn:partialSecular}, up 
to a negligible error. A related discussion is provided in 
Ref.~\cite{mccauley2019completely}. As a matter of fact, the Bloch-Redfield 
master equation does follow the dynamics of a GKLS master equation up to an 
error due to the approximation of the dynamics of the microscopic model to a 
Markovian evolution. A significant deviation of the Bloch-Redfield master 
equation from the Lindblad form must be considered as a signature of the failure 
of the Born-Markov approximations to describe the physical model, and not 
viceversa, as proven in Ref.~\cite{Hartmann2020}.

{For the reasons explained above, we are allowed to assume that the master equation in PSA with Liouvillian Eq.~\eqref{eqn:Liouvillian} can be rewritten in the GKLS or Lindblad form as in Eq.~\eqref{eqn:LindbladForm}. The Lindblad operators become linear combinations of the jump operators $\hat{A}_\alpha$, and can be obtained for each specific case by diagonalizing the matrix $\gamma_{\alpha\beta}(\omega,\omega')$ in Eq.~\eqref{eqn:dissipator} \cite{Cresser2017a,Farina2019}.} Analogously, we can write the master equation in the Lindblad form in the Heisenberg picture \cite{BreuerPetruccione}:
\begin{equation}
\label{eqn:LindbladFormHeisenberg}
\begin{split}
\frac{d}{dt}\hat{J}(t)=&i[\hat{H'},\hat{J}(t)]\\
&+\sum_{l=1}^{N^2-1} \hat{F}_l^\dagger \hat{J}(t) \hat{F}_l-\frac{1}{2}\{ 
\hat{F}_l^\dagger \hat{F}_l, J(t)\},
\end{split}
\end{equation}
where $\hat{J}=\hat{J}^\dagger$ is an observable living in $\mathsf{L}$, whose 
expectation value can be found as $\langle 
\hat{J}(t)\rangle=\Tr[\hat{J}(t)\rho_S]=\Tr[\hat{J}\rho_S(t)]$.

\subsection{Working with superoperators}
\label{sec:superoperatorsSpace}
It is very convenient to extend the bra-ket notation to the Liouville space 
$\mathsf{L}$ \cite{Albert2014,Bellomo2017}. Suppose that $\{\ket{e_j}\}_{j=1}^N$ 
is a basis of the Hilbert space $\mathsf{H}_S$. Then, any operator $\hat{O}$ (or 
equivalently density matrix) in $\mathsf{L}$ can be written as:
\begin{equation}
\label{eqn:state}
\hat{O}=\sum_{j,k=1}^NO_{jk}\ket{e_j}\bra{e_k}.
\end{equation}  
We now perform the following isomorphism, passing from a description of 
$\hat{O}$ as an operator acting on $\mathsf{H}_S$ to a description as a 
$N^2$-dimensional vector:
\begin{equation}
\label{eqn:isomorphism}
\hat{O}\rightarrow \vert O\rangle\!\rangle =\sum_{j,k=1}^N 
O_{jk}\ket{e_j}\otimes\ket{e_k}.
\end{equation}
Given $\hat{O},\hat{R}\in\mathsf{L}$, the reader can verify that this 
$N^2-$dimensional space is furnished with the Hilbert-Schmidt scalar product $ 
\langle\!\langle O\vert R\rangle\!\rangle=\Tr(\hat{O}^\dagger \hat{R})$, and 
that the following properties hold:
\begin{equation}
\label{eqn:propLiouv}
\vert O R \rangle\!\rangle=\hat{O}\otimes\mathbb{I}_N\vert 
R\rangle\!\rangle,\qquad \vert R O \rangle\!\rangle=\mathbb{I}_N\otimes 
\hat{O}^T\vert R\rangle\!\rangle,
\end{equation}
where $\mathbb{I}_N$ is the $N\times N$ identity matrix.

Using Eq.~\eqref{eqn:propLiouv}, we can now write the explicit form of the 
Liouvillian superoperator starting from the Bloch-Redfield master equation in 
Eqs.~\eqref{eqn:masterEqLiouvillian},~\eqref{eqn:lambShift} 
and~\eqref{eqn:dissipator}:
\begin{equation}
\label{eqn:LiouvillianBraket}
\begin{split}
\mathcal{L}=&-i\left( 
(\hat{H}_S+\hat{H}_{LS})\otimes\mathbb{I}_N-\mathbb{I}_N\otimes 
(\hat{H}_S+\hat{H}_{LS})^T\right)\\
&+\sum_{\alpha,\beta}\sum_{(\omega,\omega')\in\mathsf{PSA}} 
\gamma_{\alpha\beta}(\omega,\omega')\Big(\hat{A}_\beta(\omega)\otimes 
\hat{A}_\alpha^*(\omega')\Big.\\
&\Big.-\frac{1}{2}\big(\hat{A}_\alpha^\dagger(\omega')\hat{A}
_\beta(\omega)\otimes\mathbb{I}_N+\mathbb{I}_N\otimes 
(\hat{A}_\alpha^\dagger(\omega')\hat{A}_\beta(\omega))^T\big)\Big).
\end{split}
\end{equation}

{
\section{Proof of Proposition 1}
\label{sec:proofT1}
We want to proof that $[\mathcal{N},\mathcal{L}]=0$. For convenience, we work using the isomorphism ``flattening'' matrices into vectors (see 
Appendix~\ref{sec:superoperatorsSpace}, Eq.~\eqref{eqn:isomorphism}), so that $\mathcal{L}$ can be written as in Eq.~\eqref{eqn:LiouvillianBraket} and $\mathcal{N}$ as in Eq.~\eqref{eqn:superoperatorN}. 
Looking at the structure of $\mathcal{L}$, we recognize four different parts of the Liouvillian that must commute with $\mathcal{N}$; in particular, to proof the statement it is sufficient to verify the following four assertions:
\begin{enumerate}
\item [(1)]$[\mathcal{N},\hat{H}_S\otimes\mathbb{I}_N]=[\mathcal{N},\mathbb{I}_N\otimes \hat{H}_S^T]=0$.
\item [(2)]$[\mathcal{N},\hat{A}_\beta(\omega)\otimes 
\hat{A}_\alpha^*(\omega')]=0$ for all $\alpha,\beta$ and $(\omega,\omega')\in\mathsf{PSA}$.
\item [(3)]$[\mathcal{N},\hat{A}_\alpha^\dagger(\omega')\hat{A}
_\beta(\omega)\otimes\mathbb{I}_N]=0$ for all $\alpha,\beta$ and $(\omega,\omega')\in\mathsf{PSA}$.
\item [(4)]$[\mathcal{N},\mathbb{I}_N\otimes 
(\hat{A}_\alpha^\dagger(\omega')\hat{A}_\beta(\omega))^T]=0$ for all $\alpha,\beta$ and $(\omega,\omega')\in\mathsf{PSA}$.
\end{enumerate} 
This means that the value of the coefficients of the master equation does not play a role in the appearance of the symmetry.

Assertion (1) is easily proven: the system Hamiltonian 
{Eq.~\eqref{eqn:HamiltonianNonIntModes}} cannot change the number of 
particles in any mode, since $[\hat{H}_S,\hat{n}_k]=0\;\forall\,k$, and thus 
$[\hat{H}_S,\hat{N}]=0$, therefore $[\mathcal{N},\hat{H}_S\otimes\mathbb{I}_N]=0$.

Given the Hamiltonian of a system of $M$ modes, $\hat{H}_S=\sum_{k=1}^M E_k \hat{n}_k$, we write an eigenvector as $\ket{e}$, with $\hat{H}_S\ket{e}=e\ket{e}$ and $e=\sum_{k=1}^M E_k n_k^e$, where $n_k^e$ is the number of excitation in each mode of $\ket{e}=\ket{n_1^e,\ldots,n_M^e}$. The total number of particles in $\ket{e}$ is given by $n^e=\sum_{k=1}^M n_k^e$. Using the same notation for generic eigenvectors $\ket{e'},\ket{\epsilon},\ket{\epsilon'}$, we write the jump operators as $\hat{A}_\beta(\omega)=\sum_{e'-e=\omega} \ket{e}\mel{e}{\hat{A}_\beta}{e'}\bra{e'}$, $\hat{A}_\alpha(\omega')=\sum_{\epsilon'-\epsilon=\omega'} \ket{\epsilon}\mel{\epsilon}{\hat{A}_\alpha}{\epsilon'}\bra{\epsilon'}$. Let us now consider Assertion (2): we write the commutator as
\begin{equation}
\label{eqn:assertion2}
\begin{split}
&[\mathcal{N},\hat{A}_\beta(\omega)\otimes 
\hat{A}_\alpha^*(\omega')]\\
=&\hat{N}\hat{A}_\beta(\omega)\otimes\hat{A}_\alpha^*(\omega')-\hat{A}_\beta(\omega)\hat{N}\otimes\hat{A}_\alpha^*(\omega')\\
&+\hat{A}_\beta(\omega)\otimes\hat{A}_\alpha^*(\omega')\hat{N}^T-\hat{A}_\beta(\omega)\otimes\hat{N}^T\hat{A}_\alpha^*(\omega')\\
=&\sum_{e'-e=\omega} (n^e-n^{e'})\ket{e}\mel{e}{\hat{A}_\beta}{e'}\bra{e'}\otimes\hat{A}_\alpha^*(\omega')\\
&-\hat{A}_\beta(\omega)\otimes\sum_{\epsilon'-\epsilon=\omega'} (n^{\epsilon}-n^{\epsilon'})\ket{\epsilon}\mel{\epsilon}{\hat{A}_\alpha^*}{\epsilon'}\bra{\epsilon'}.
\end{split}
\end{equation}
Since $(\omega,\omega')\in\mathsf{PSA}$, according to Eq.~\eqref{eqn:partialSecular} we must have $\omega-\omega'=\mathcal{O}_{t^*}(\mu^2)$, therefore $\sum_{k=1}^M E_k (n_k^{e'}-n_k^e-n_k^{\epsilon'}+n_k^{\epsilon})=\mathcal{O}_{t^*}(\mu^2)$ or $\sum_{k=1}^M E_k (n_k^{e'}+n_k^{\epsilon})=\sum_{k=1}^M E_k (n_k^e+n_k^{\epsilon'})+\mathcal{O}_{t^*}(\mu^2)$. But, assuming \textit{Condition II} on the interaction Hamiltonian defined at the end of Sec.~\ref{sec:requirements} (which comprises \textit{Condition I} as well), the last line means that $\sum_{k=1}^M (n_k^{e'}+n_k^{\epsilon})=\sum_{k=1}^M (n_k^e+n_k^{\epsilon'})$, that is to say, $(n^e-n^{e'})=(n^{\epsilon}-n^{\epsilon'})$, for any couple of $\ket{e},\ket{e'}$ or $\ket{\epsilon},\ket{\epsilon'}$ in Eq.~\eqref{eqn:assertion2}. Therefore, the commutator in Eq.~\eqref{eqn:assertion2} vanishes and we have proven Assertion (2). This proof shows us how we can relax \textit{Condition II} on the interaction Hamiltonian: it is sufficient to assume it only on the energies which enter in the expression of each possible $(\omega,\omega')\in\mathsf{PSA}$. For instance, suppose we 
have a system of two modes and all the $\hat{A}_\alpha$ are second-degree polynomials in 
the creation and annihilation operators of the modes. Then, we just need to 
require that $2 E_1\neq E_2+ \mathcal{O}_{t^*}(\mu^2)$ or viceversa, in order to eliminate all 
the ``unbalanced'' terms through the partial secular approximation.

Next, we consider Assertion (3):
\begin{equation}
\label{eqn:assertion3}
\begin{split}
&[\mathcal{N},\hat{A}_\alpha^\dagger(\omega')\hat{A}
_\beta(\omega)\otimes\mathbb{I}_N]\\
=&\hat{N}\hat{A}_\alpha^\dagger(\omega')\hat{A}
_\beta(\omega)\otimes\mathbb{I}_N-\hat{A}_\alpha^\dagger(\omega')\hat{A}
_\beta(\omega)\hat{N}\otimes\mathbb{I}_N\\
=&\sum_{\substack{\epsilon-\epsilon'=\omega'\\
                  e-\epsilon'=\omega}} (n^{\epsilon}-n^{e})\ket{\epsilon}\mel{\epsilon}{\hat{A}_\alpha^\dagger}{\epsilon'}\mel{\epsilon'}{\hat{A}_\beta}{e}\bra{e}\otimes\mathbb{I}_N.
\end{split}
\end{equation}
Applying \textit{Condition II} on the energy difference $\omega-\omega'$ as for Assertion (2), we find that $n^\epsilon=n^e$ and the commutator in Eq.~\eqref{eqn:assertion3} vanishes, proving Assertion (3). Assertion (4) is verified analogously, and we have proven Proposition 1.
}

{\section{Jordan-Wigner transformations}
\label{sec:JordanWigner}
In this appendix we briefly present the well-known Jordan-Wigner 
technique \cite{Lieb1961,coleman2015introduction} to represent spins as fermions, and we show how to employ it to recast the Hamiltonian of an excitation-preserving spin chain in a quadratic fermionic Hamiltonian. Given a system of $M$ spins, we apply the following Jordan-Wigner transformations to write each spin operator as a function of fermionic operators:
\begin{equation}
\begin{aligned}
\label{eqn:JordanWignerGen}
&\sigma_k^z=c_k^\dagger c_k-\frac{1}{2},\\
&\sigma_k^+=c_k^\dagger \,e^{i\pi\sum_{l<k} n_l},\\
&\sigma_k^-=c_k \, e^{-i\pi\sum_{l<k} n_l}.
\end{aligned}
\end{equation}
The reader can verify the anticommutation rules 
$\{c_j,c_k^\dagger\}=\delta_{jk}$, $\{c_j,c_k\}=0$.

Let us now suppose that the spins are interacting in an excitation-preserving chain, with Hamiltonian:
\begin{equation}
\label{eqn:HamChain}
H_{SC}=\sum_{k=1}^M \frac{\omega_k}{2}\sigma_k^z+\sum_{k=1}^{M-1} J_k (\sigma_{k+1}^+\sigma_k^-+h.c.).
\end{equation}
Using the Jordan-Wigner transformations in Eq.~\eqref{eqn:JordanWignerGen}, we can express it as:
\begin{equation}
\label{eqn:HamChainJW}
\begin{split}
H_{SC}=&\sum_{k=1}^M \frac{\omega_k}{2}\left(c_k^\dagger c_k-\frac{1}{2}\right)\\
&+\sum_{k=1}^{M-1} J_k (c_{k+1}^\dagger e^{i\pi n_k} c_k^-+h.c.),
\end{split}
\end{equation}
where we have used $[e^{i\pi\sum_{l<k} n_l},c_k]=0$ for $k\geq l$. Noticing that $e^{i\pi n_k}c_k\ket{0}_k=0$ and $e^{i\pi n_k}c_k\ket{1}_k=c_k\ket{1}_k$, we observe that the phase $e^{i\pi n_k}$ has no effects and can be removed from the Hamiltonian. Therefore, Eq.~\eqref{eqn:HamChainJW} is quadratic in the fermionic operators and can be recast in the form of Eq.~\eqref{eqn:HamiltonianNonIntModes}, thus being suitable for the analysis in Sec.~\ref{sec:block}.
}

\section{Two coupled spins as free fermions}
\label{sec:coupledSpinsFree}
{In this appendix we show how to employ the Jordan-Wigner transformations to write a system of two coupled spins as non-interacting fermions (part of the discussion was already addressed in 
Ref.~\cite{Giorgi2016}). For convenience, we rewrite the Jordan-Wigner transformations Eq.~\eqref{eqn:JordanWignerGen} for two spins as:
\begin{equation}
\begin{aligned}
\label{eqn:JordanWignerSub1}
&\sigma_1^z=1-2 c_1^\dagger c_1,\quad& \sigma_2^z&=1-2 c_2^\dagger c_2,\\
&\sigma_1^x=c_1^\dagger+c_1,\quad &\sigma_2^x&=(1-2 c_1^\dagger c_1)(c_2^\dagger+c_2),
\end{aligned}
\end{equation}
where $c_1$ and $c_2$ are fermionic operators. The free Hamiltonian 
of the coupled qubits, given in Eq.~\eqref{eqn:HamiltonianSpins}, is now transformed following Eq.~\eqref{eqn:JordanWignerSub1}:
\begin{equation}
\label{eqn:freeHamJW}
\begin{split}
H_S=&\frac{\omega_1}{2}(1-2 c_1^\dagger c_1)+\frac{\omega_2}{2}(1-2 c_2^\dagger 
c_2)\\
&+\lambda(c_1^\dagger-c_1)(c_2^\dagger+c_2).
\end{split}
\end{equation}}

In order to diagonalize $H_S$ written in terms of fermionic operators, we first 
perform the Bogoliubov transformation
\begin{equation}
\label{eqn:Bogoliubov}
\begin{aligned}
&c_1=\cos\theta \,\xi_1+\sin\theta \,\xi_2^\dagger, \\
& c_2=\cos\theta \,\xi_2-\sin\theta\, \xi_1^\dagger,
\end{aligned}
\end{equation}
and then the rotation
\begin{equation}
\label{eqn:rotation}
\begin{aligned}
&\xi_1=\cos\phi\,f_1^\dagger+\sin\phi\, f_2^\dagger,\\
&\xi_2=\cos\phi\,f_2^\dagger-\sin\phi \,f_1^\dagger.
\end{aligned}
\end{equation}
\begin{widetext}
Let us now write Eq.~\ref{eqn:freeHamJW} after having applied the Bogoliubov 
transformation:
\begin{equation}
\label{eqn:passage1}
\begin{split}
H_S=&+\frac{\omega_1}{2}\Big[1-2\Big(\cos^2\theta\, 
\xi_1^\dagger\xi_1+\sin^2\theta \,\xi_2\xi_2^\dagger+\sin\theta\cos\theta 
(\xi_2\xi_1+h.c.)\big)\Big]\\
&+\frac{\omega_2}{2}\Big[1-2\big(\cos^2\theta\, \xi_2^\dagger\xi_2+\sin^2\theta 
\,\xi_1\xi_1^\dagger+\sin\theta\cos\theta (\xi_2\xi_1+h.c.)\big)\Big]\\
&+\lambda \Big[2\cos\theta\sin\theta 
(\xi_1\xi_1^\dagger+\xi_2\xi_2^\dagger)-2\cos\theta\sin\theta+(\cos^2\theta-\sin
^2\theta)(\xi_2\xi_1+h.c.)+(\xi_1^\dagger\xi_2+h.c.) \Big].
\end{split}
\end{equation}
We set $\theta$ so as to delete all the double-excitation terms in 
Eq.~\ref{eqn:passage1}:
\begin{equation}
\label{eqn:setTheta}
\begin{split}
&-\omega_+\sin\theta\cos\theta+\lambda(\cos^2\theta-\sin^2\theta)=0\qquad\Rightarrow \qquad\tan{2\theta}=\frac{2\lambda}{\omega_+}
\end{split}
\end{equation}
with $\omega_+=\omega_1+\omega_2$.

Using the condition in Eq.~\ref{eqn:setTheta}, we now write $H_S$ after having 
performed the rotation:
\begin{equation}
\label{eqn:passage2}
\begin{split}
H_S=&+\frac{\omega_1}{2}\left[
1-2\left((\sin^2\theta\sin^2\phi-\cos^2\theta\cos^2\phi) f_1^\dagger 
f_1\right.\right.+(\sin^2\theta\cos^2\phi-\cos^2\theta\sin^2\phi)f_2^\dagger 
f_2\\
&\qquad\quad\left.\left.+\sin\phi\cos\phi(f_1 
f_2^\dagger+h.c.)+\cos^2\theta\right)\right]\\
&+\frac{\omega_2}{2}\left[
1-2\left((\sin^2\theta\sin^2\phi-\cos^2\theta\cos^2\phi) f_2^\dagger 
f_2\right.\right.+(\sin^2\theta\cos^2\phi-\cos^2\theta\sin^2\phi) f_1^\dagger 
f_1\\
&\qquad\quad\left.\left.-\sin\phi\cos\phi(f_1 
f_2^\dagger+h.c.)+\cos^2\theta\right)\right]\\
&+\lambda\left[\sin{2\theta}(f_1^\dagger f_1+f_2^\dagger 
f_2)+\sin{2\phi}(f_1^\dagger f_1-f_2^\dagger f_2)\right.
+(\cos^2\phi-\sin^2\phi)(f_1 f_2^\dagger+h.c.)\left.-\sin{2\theta}\right].
\end{split}
\end{equation}
In order to eliminate the remaining cross terms, we set the value of $\phi$:
\begin{equation}
\label{eqn:setPhi}
\begin{split}
&-\omega_-\sin\phi\cos\phi+\lambda(\cos^2\phi-\sin^2\phi)=0\quad\Rightarrow\quad
{\tan{2\phi}=\frac{2\lambda}{\omega_-}}
\end{split}
\end{equation}
with $\omega_-=\omega_1-\omega_2$.

By employing the relations 
$\cos^2\alpha\cos^2\beta-\sin^2\alpha\sin^2\beta=(\cos{2\alpha}+\cos{2\beta})/2$ 
and 
$\cos^2\alpha\sin^2\beta-\sin^2\alpha\cos^2\beta=(\cos{2\alpha}-\cos{2\beta})/2$
, we finally obtain the Hamiltonian
\begin{equation}
\label{eqn:HSfinalJW}
H_S=E_1\left(2 f_1^\dagger f_1-1\right)+E_2\left(2 f_2^\dagger f_2-1\right),
\end{equation}
where 
\begin{equation}
\label{eqn:E1E2}
\begin{split}
&E_1=\frac{\sqrt{\lambda^2+\omega_+^2/4}+\sqrt{\lambda^2+\omega_-^2/4}}{2}, 
\qquad \quad 
E_2=\frac{\sqrt{\lambda^2+\omega_+^2/4}-\sqrt{\lambda^2+\omega_-^2/4}}{2}.
\end{split}
\end{equation}

We can now proceed to write the spin operators in terms of the fermionic 
operators. Let us start with $\sigma_1^x$:
\begin{equation}
\label{eqn:sigma1xJW}
\sigma_1^x=\cos(\theta+\phi)\left(f_1^\dagger+f_1\right)+\sin(\theta+\phi)\left(
f_2^\dagger+f_2\right).
\end{equation}
By noticing that $(1-2c_1^\dagger c_1)(1-2c_2^\dagger 
c_2)(c_2-c_2^\dagger)=\sigma_2^x$, we can readily obtain
\begin{equation}
\label{eqn:sigma2xJW}
\sigma_2^x=\cos(\theta-\phi)P\left(f_2^\dagger-f_2\right)+\sin(\theta-\phi)P\left(f_1^\dagger-f_1\right),
\end{equation}
where 
\begin{equation}
\label{eqn:parity}
P=(1-2c_1^\dagger c_1)(1-2c_2^\dagger c_2)=(2f_1^\dagger f_1-1)(2f_2^\dagger 
f_2-1)
\end{equation}
is the parity operator, which tells us whether the number of excitations in the 
system is even or odd. The Hamiltonian $H_S$ conserves the parity of the 
excitation number, i.e. $[H_S,P]=0$, thus we are sure that ${P}$ has the form 
presented in Eq.~\ref{eqn:parity}.

The form of the operators $\sigma_1^z$ and $\sigma_2^z$ is more involved, since 
they inevitably contain the ``double emission'' and ``double absorption'' terms 
$f_1 f_2$ and $f_1^\dagger f_2^\dagger$, which we could find in the coupling of 
the original Hamiltonian Eq.~\ref{eqn:HamiltonianSpins} 
$\lambda\sigma_1^x\sigma_2^x$. In some particular scenarios, it is possible to 
perform a rotating wave approximation on such direct interaction, and to write 
it as $\lambda(\sigma_1^+\sigma_2^-+\sigma_1^-\sigma_2^+)$, which does not add 
excitations into the system. In this case, diagonalizing the system Hamiltonian 
is easier and can be done by just a single rotation \cite{Rivas2010a}. 
Anyway, with the aim at a more complete description, we keep the 
counter-rotating terms in the Hamiltonian and we write the operators as:
\begin{equation}
\label{eqn:sigmaZjw}
\begin{aligned}
\sigma_1^z=&(\cos{2\theta}+\cos{2\phi})f_1^\dagger 
f_1+(\cos{2\theta}-\cos{2\phi})f_2^\dagger 
f_2-\cos{2\theta}-2\left[
\cos\phi\sin\phi(f_1f_2^\dagger+h.c.)+\cos\theta\sin\theta(f_1f_2+h.c.)\right]
.\\
\sigma_2^z=&(\cos{2\theta}+\cos{2\phi})f_2^\dagger 
f_2+(\cos{2\theta}-\cos{2\phi})f_1^\dagger 
f_1-\cos{2\theta}-2\left[\cos\phi\sin\phi(f_1^\dagger 
f_2+h.c.)+\cos\theta\sin\theta(f_1f_2+h.c.)\right].
\end{aligned}
\end{equation}
\end{widetext}
Finally, we find the new basis that diagonalizes $H_S$ as a function of the 
canonical basis $\{\ket{11},\ket{10},\ket{01},\ket{00}\}$, which corresponds 
respectively to both spins up, first spin up and second down, etc... To 
represent the excitation basis of each couple of fermionic operators, we employ 
a subscript indicating to which operator we are referring, while we do not use 
subscripts for the canonical basis; for instance from 
Eq.~\ref{eqn:JordanWignerSub1} we understand that:
\begin{equation}
\label{eqn:canVsC}
\begin{split}
&\ket{00}_c=\ket{11},\qquad \ket{01}_c=\ket{10},\\
& \ket{10}_c=\ket{01},\qquad \ket{11}_c=\ket{00}.
\end{split}
\end{equation}

From Eq.~\ref{eqn:rotation}, we see that the vacuum state of $f_1,f_2$ is the 
fully-excited state of $\xi_1,\xi_2$, i.e. $\ket{00}_f=\ket{11}_\xi$. In order 
to find $\ket{00}_f$, i.e. the ground state of $H_S$, we thus impose that 
$\xi_1^\dagger$ and $\xi_2^\dagger$ applied on a linear combination of the 
states in Eq.~\ref{eqn:canVsC} read $0$. For instance,
\begin{equation}
\begin{split}
&\xi_1^\dagger\sum_{\alpha,\beta=0,1} a_{\alpha\beta}\ket{\alpha\beta}_c=0\\
&\Rightarrow \cos\theta\, a_{00}=-\sin\theta\,a_{11},\qquad a_{01}=0.
\end{split}
\end{equation}
Finally we have:
\begin{equation}
\label{eqn:groundState}
\ket{00}_f=+\sin\theta\ket{11}-\cos\theta\ket{00}.
\end{equation}
The remaining states are obtained by applying $f_1^\dagger$ and $f_2^\dagger$ on 
the ground state, and they read:
\begin{equation}
\label{eqn:remainingStates}
\begin{aligned}
&\ket{01}_f=-\sin\phi\ket{10}+\cos\phi\ket{01},\\
&\ket{10}_f=-\cos\phi\ket{10}-\sin\phi\ket{01},\\
&\ket{11}_f=+\cos\theta\ket{11}+\sin\theta\ket{00}.
\end{aligned}
\end{equation}

\bibliography{blockStructure}

\end{document}